\documentclass[a4paper,UKenglish,cleveref,autoref,thm-restate,final]{lipics-v2021}

\hideLIPIcs  

\title{A Zielonka-type Construction for Connectedly Communicating Processes} 

\author{Béatrice Bérard}{Sorbonne Université, CNRS, LIP6, F-75005 Paris, France}{beatrice.berard@lip6.fr}{}{}
\author{Benjamin Monmege}{Aix-Marseille Univ, CNRS, LIS, Marseille, France}{benjamin.monmege@univ-amu.fr}{}{}
\author{B Srivathsan}{Chennai Mathematical Institute, Chennai, India}{sri@cmi.ac.in}{}{}
\author{Arnab Sur}{Chennai Mathematical Institute, Chennai, India}{arnabs@cmi.ac.in}{}{}

\authorrunning{B.~B\'erard, B.~Monmege, B.~Srivathsan, and A.~Sur}

\Copyright{Béatrice B\'erard, Benjamin Monmege, B.~Srivathsan, and Arnab~Sur} 



\nolinenumbers 

\EventEditors{John Q. Open and Joan R. Access}
\EventNoEds{2}
\EventLongTitle{42nd Conference on Very Important Topics (CVIT 2016)}
\EventShortTitle{CVIT 2016}
\EventAcronym{CVIT}
\EventYear{2016}
\EventDate{December 24--27, 2016}
\EventLocation{Little Whinging, United Kingdom}
\EventLogo{}
\SeriesVolume{42}
\ArticleNo{23}

\usepackage{tikz}
\usetikzlibrary{positioning,automata}
\usepackage[table]{xcolor}

\newcommand\N{\mathbb N}

\newcommand\loc{\mathrm{loc}}
\newcommand\events{\mathcal E}
\newcommand\Foata{\mathsf F}
\newcommand\Foatalength[1]{|#1|_{\mathsf F}}
\newcommand\view{\mathrm{view}}

\newcommand\incl{\subseteq}
\newcommand\diam{\mathrm{diam}}
\newcommand\A{\mathcal A}
\newcommand\B{\mathcal B}
\newcommand\ideal[1]{#1{\downarrow}}

\newcommand\length[1]{|#1|}

\newcommand{\synchronise}{\mathsf{synchronise}}
\newcommand{\expand}{\mathsf{expand}}
\newcommand{\cut}{\mathsf{cut}}

\newcommand{\proc}{P}
 
\newcommand{\xra}{\xrightarrow}
\newcommand{\blue}[1]{\textcolor{blue}{#1}}

\newcommand{\dvdcut}{\mathsf{cut\textsf{-}split}}

\newcommand\Lt{L_{\mathsf{tr}}}

\newcommand{\sock}{\mathsf{SOCK}}
\newcommand{\dvline}{%
  \tikz[baseline={(0,-0.5ex)}]
    \draw[dashed] (0,-0.5) -- (0,0.5);%
}

\usepackage[obeyFinal]{todonotes}
\begin{document}

\maketitle

\begin{abstract}
    Given a global specification as a trace-closed regular language, Zielonka's theorem provides a construction to synthesise a language equivalent distributed implementation represented as a deterministic asynchronous automaton (AA). The construction is notoriously complicated, which has led to a line of work that considers restrictions on the specifications or on the distributed architectures, with the objective of providing a conceptually simpler construction. A new construction has recently been provided for ``fair'' specifications, in which all processes participate regularly. 
In this work, we enhance this construction to enable deterministic finite-state automata (DFA) specifications with ``connectedly communicating processes'': there should be a constant delay~$d$ such that if two processes do not hear from one another after this delay, they will never hear from one another until the end of the run. This is a relaxation of the fairness constraint, in which some process may deliberately stop communicating with another one, e.g.~a client-server architecture, where a client stops asking the server for a resource if it did not get any response from it after a while.
Our construction results in an AA
where every process has a number of local states that is polynomial in the
number of states of the DFA, and where the only exponential explosion
is related to the parameter $d$, and the separation depth of processes. 
\end{abstract}

\section{Introduction}

Asynchronous automata (AA) are a foundational model for concurrent systems, where independent processes synchronise on a predefined set of shared actions. The property of connectedly communicating processes was introduced in~\cite{MadThi05} as a subclass of AAs where the monadic second order logic theory is decidable in the branching time setting. An AA is said to be \emph{connectedly communicating} if there is a bound $d$ such that, when some process executes $d$ steps without hearing from some other process, then these two processes will not communicate any more in the future, even indirectly. In the context of distributed algorithms, this property can model a permanent crash of a communication link, which is a realistic assumption in large networks. Some algorithms like consensus or election are designed to be (sometimes partially) resilient with respect to such failures~\cite{Weiss2001ConsensusWW,Biswas2021election}. In the asynchronous communication setting, when there is no uniform bound on communication delays, a process which does not receive messages from another process may also assume that a failure occurred and stop communicating with this process. This would be the case for instance in a client-server architecture, where the client stops sending requests if it does not receive any answer from the server for some time. 

Zielonka's theorem provides a method for synthesising a deterministic AA from a given deterministic finite-state automaton (DFA) satisfying a certain commutativity property, and an alphabet distribution among processes~\cite{Zie87}. This construction is known to be extremely complex and has been the subject of numerous refinements and optimisations (e.g.~\cite{CorMet93,MukSoh97,Genest06}), leading to an optimal Zielonka-type construction~\cite{GenGim10}. In all these constructions, each process keeps track of what it believes to be the latest ``information'' available to every other process. When there is a shared action, the processes share their local information, reconcile them and update their local states. This underlying idea is implemented using a \emph{gossip automaton} in~\cite{MukSoh97,Madhavan-Notes-Zielonka}, and so-called \emph{zones} in~\cite{Genest06,GenGim10}. This yields a number of local states that is exponential in the number of processes, whereas the construction can be kept polynomial in the number of states of the~DFA. 

Several solutions have been considered to avoid the resolution of this \emph{gossip problem}, e.g.~compositional methods~\cite{Bau11,pighizzini1993synthesis,baudru2006unfolding} that still require a number of states exponential in the number of processes. Other solutions are based on restricting the topology of communications~\cite{KRISHNA2013109,HauLeh24} and 
a recent work~\cite{AdsGas24} has proposed a logic-based construction going through a local, past-oriented fragment of propositional dynamic logic. The AA is obtained by a cascade product of localised AAs, which essentially operate on a single process. 

Finally, a recent proposal~\cite{BerMon26} considered so-called \emph{fair} specifications: a DFA is $d$-fair whenever every word
of length $d$ that can be read in the DFA makes every process participate at least once. In this case, a synthesis 
procedure is presented where the resulting AA is linear in the size of the DFA, polynomial in the size of the alphabet, and the only exponential explosion is related to the parameter $d$. The complexity is independent of the number of processes. The synthesis procedure has been implemented in a tool FAAST~\cite{10.5281/zenodo.18165965} and has also been used to synthesise AAs recognising all $d$-fair traces of a (potentially non-fair) DFA specification.
Other orthogonal restrictions on the DFA specification leading to implementations include safe AAs~\cite{SteEsp03} and realistic AAs~\cite{AksDin13}.

Our goal in this paper is to present an AA synthesis procedure for specifications corresponding to connectedly communicating processes. The condition of connected communication that we study in this article can be seen as a relaxation of the fairness constraint: we first show that in DFA specifications modelling connectedly communicating processes, the strongly connected components are fair for the subset of processes that are still potentially communicating with each other. Then, we adapt the construction proposed in~\cite{BerMon26}, wherein each process maintains a finite suffix of its view. Thanks to the fairness hypothesis, the prefix of its view that it has cut is known by every other process. We cannot do this any more for our case of connected communication, since processes may separate after having communicated to a certain point. To handle this aspect, we need a mechanism to cut parts of the view in the middle (and not necessarily a prefix). Hence, our modified construction splits the view into \emph{phases} and cuts a prefix within each phase. 
The phase splits happen when some processes stop communicating, and thus do not occur too often. Therefore, we obtain a construction that is polynomial in the
number of states of the DFA, and where the only exponential explosion
is related to the parameter $d$, and the separation depth of processes. 

After presenting preliminary notions in Section~\ref{sec:prelim}, we recall the construction for fair systems~\cite{BerMon26} in Section~\ref{sec:fair-recall}. In Section~\ref{sec:ccp-basic}, we provide a characterization of DFAs corresponding to connectedly communicating processes and go on to the full construction in Section~\ref{sec:ccp-construction}. Additional examples are presented in Appendix~\ref{app:additional-examples}. Missing proofs and technical details can be found in other clearly marked appendices.

\section{Preliminaries}\label{sec:prelim}

We build upon the work in~\cite{BerMon26}, and thus reuse some of the preliminary definitions as in~\cite{BerMon26}.
%
Let $\Sigma$ be a finite alphabet. A \emph{concurrent alphabet} is a pair $(\Sigma, I)$ with $I\subseteq \Sigma\times \Sigma$ a symmetric and irreflexive relation called the \emph{independence relation}. The corresponding dependence relation is the set $D = (\Sigma\times \Sigma) \setminus I$. 
A concrete independence relation can be obtained by distributing the letters into a finite set $P$ of processes, formally defined by a function $\loc\colon \Sigma \to 2^P$ where $\loc(a)$ is the set of processes that can read the letter $a \in \Sigma$. We further define $\Sigma_p = \{a \in \Sigma \mid p \in \loc(a)\}$ as the alphabet of process $p\in P$. We call $(\Sigma, \loc)$ a \emph{distributed alphabet}. Such a distribution naturally leads to an independence relation $I_\loc = \{(a, b) \in \Sigma\times \Sigma \mid \loc(a) \cap \loc(b) = \emptyset\}$.
We extend the function $\loc$ to words by letting $\loc(\varepsilon) = \emptyset$, and $\loc(ua) = \loc(u) \cup \loc(a)$ for all $u \in \Sigma^*$ and $a \in \Sigma$.

\begin{example}\label{ex:distributed-alphabet} In the examples of this article, we will describe letters of the alphabet using the set of processes that synchronise on them -- always supposing that there is at most one action used by any subset of processes to synchronise. For instance, consider $P= \{p_1, p_2, p_3,p_4, p_5\}$ to be the processes. For our running example, we will use the alphabet $\Sigma = \{12, 13, 2, 235, 24, 34, 35, 4\}$. The distribution is given by $\loc(i_1 \dots i_\ell) = \{p_{i_1}, \dots p_{i_\ell}\}$ for each $i_1 \dots i_\ell \in \Sigma^+$ with each $i_j \in \{1, 2, \dots 5\}$. As an illustration, $\loc(12) = \{p_1, p_2\}$, and $\Sigma_{p_1} = \{12, 13\}$. The independence relation $I_\loc$ contains $(12, 34)$, and $(2, 35)$ for instance, but not $(12, 13)$ since process $p_1$ participates in both actions.  
\end{example}



\subparagraph*{Traces} A \emph{trace} over a concurrent alphabet $(\Sigma, I)$ is a labelled partial order $t = (\events, \leq, \lambda)$ where $\events$ is a set of \emph{events}, $\lambda\colon \events \to \Sigma$ labels each event by a letter, and $\leq$ is a partial order of $\events$ such that for all $e, f \in \events$, 
$(\lambda(e), \lambda(f)) \notin I$ implies $e\leq f$ or $f\leq e$; and secondly, 
$e \lessdot f$ implies $(\lambda(e), \lambda(f)) \notin I$ where ${\lessdot} =  {\le} \setminus {\le^2}$ is the immediate successor relation induced by the partial order. 
The number of events in the trace $t$ is denoted as $|t|$, and is called its \emph{length}. If we have a distributed alphabet $(\Sigma, \loc)$, we also let $|t|_p$ be the number of events in which a process $p$ participated, i.e.~$|t|_p = |\{e\in \events \mid \lambda(e)\in \loc(p)\}|$. 

A word $w=a_0\cdots a_{n-1}\in \Sigma^*$ gives rise to a unique trace by putting one event per position in the word, and defining the successor relation $\lessdot$ as all the pairs $(i, j)$ of positions such that $i<j$, $(a_i, a_j)\notin I$ and there are no positions $k$ such that $i<k<j$ and $(a_i, a_k), (a_k, a_j)\notin I$. We call this word a \emph{linearisation} of the trace. Two words $w$ and $w'$ mapped to the same trace are said to be \emph{equivalent}. 
The lengths $|w|$ and $|w|_p$ are defined similarly as for traces.

Minimal (respectively, maximal) elements of a trace $t$ are all the events that have no smaller (respectively, larger) events. The set of minimal (respectively, maximal) events of~$t$ is denoted by $\min(t)$ (respectively, $\max(t)$). 
Given two traces $t_1 = (\events_1, \leq_1, \lambda_1)$ and $t_2 = (\events_2, \leq_2, \lambda_2)$, the concatenation $t_1 t_2$ is the trace $(\events', \leq', \lambda')$ where $\events'= \events_1 \cup \events_2$, $\lambda'(e) = \lambda_1(e)$ if $e \in \events_1$, and $\lambda'(e) = \lambda_2(e)$ otherwise, and $\leq'$ is the transitive closure of ${\leq_1} \cup {\leq_2} \cup \{ (x, y) \mid  x \in \events_1, y \in \events_2, \text{ and } (\lambda'(x), \lambda'(y)) \notin I\}$.  

\begin{example}\label{ex:trace} A trace over the distributed alphabet of Example~\ref{ex:distributed-alphabet} can be depicted 
  as follows, where the arrows between the events denote the $\lessdot$ relation:
  \begin{center}
    \scalebox{.8}
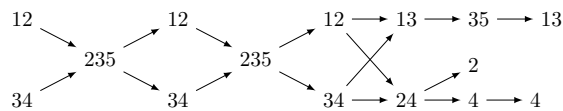
{
    \begin{tikzpicture}[>=latex,node distance=0.2cm and .6cm]
      \node (a) {$12$};
            \node[below right=of a] (c) {$235$};
            \node[below left=of c] (b) {$34$};
            \node[above right=of c](d) {$12$};
            \node[below right=of d] (f) {$235$};
            \node[below left=of f] (e) {$34$};
            \node[above right=of f](g) {$12$};
            \node[below right=of f] (h) {$34$};
            \node[right=of g] (i) {$13$};
            \node[right=of h] (j) {$24$};
            \node[right=of i] (k) {$35$};
            \node[right=of k] (n) {$13$};
            \node[right=of j] (m) {$4$};
            \node[above=0.1cm of m] (l) {$2$};
            \node[right=of m] (o) {$4$};

      \draw[->] (a) edge (c)
      (b) edge (c)
      (c) edge (d)
      (c) edge (e)
      (d) edge (f)
      (e) edge (f)
      (f) edge (g)
      (f) edge (h)
      (g) edge (i)
      (g) edge (j)
      (h) edge (i)
      (h) edge (j)
      (i) edge (k) 
      (k) edge (n)
      (j) edge (l)
      (j) edge (m)
      (m) edge (o);
    \end{tikzpicture}}
  \end{center}
   One of its linearisations is $12\quad 34 \quad 235\quad 34\quad 12\quad 235\quad 12\quad 34\quad 13\quad 35\quad 24\quad 4\quad 2\quad 13\quad 4$. For instance, the two first events can be switched to obtain another linearisation of the same trace.
   The minimal events are the ones labelled by $12$ and $34$ on the left. The maximal events are those labelled by $13$, $2$ and $4$ on the right. 
  
\end{example}



\subparagraph*{Views}
For a trace $t = (\mathcal E, {\leq}, \lambda)$, a subset $J\subseteq \mathcal E$ is called an \emph{ideal} of $t$ if for all $e\in J$, and $f\in \mathcal E$ such that $f\leq e$, we have $f\in J$. An ideal can also be seen as a trace $u$ which is 
a \emph{prefix} of the trace $t$, i.e., there exists a trace $v$ such that $t = uv$.
For a subset $S\subseteq \mathcal E$ of events, we let $\ideal{S}$ be the ideal $\{f\in \mathcal E \mid \exists e\in S \; f\leq e\}$.
We identify special ideals called \emph{views}: the \emph{view} of a process $p$ in a trace $t$ is the ideal of~$t$ consisting of all events currently known by the process $p$. Formally, we let $\max_p(t)$ be the largest event of $t$ that is labelled in $\Sigma_p$. Then, the view of process $p$ in a trace $t$, denoted as $\view_p(t)$, is the ideal $\ideal{\{\max_p(t)\}}$. 
The view of a set of processes $X \subseteq P$ in a trace $t$ is the ideal 
$\ideal{\{\max_p(t) \mid p\in X\}}$, obtained as the union of the views of the processes in $X$. 

\begin{example}\label{ex:views} Consider again the trace of Example~\ref{ex:trace}. Then, $\view_{p_2}(t)$ is equal to:
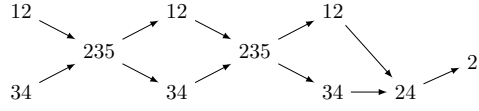
\begin{center}
    \scalebox{.8}{
    \begin{tikzpicture}[>=latex,node distance=0.2cm and .6cm]
      \node (a) {$12$};
            \node[below right=of a] (c) {$235$};
            \node[below left=of c] (b) {$34$};
            \node[above right=of c](d) {$12$};
            \node[below right=of d] (f) {$235$};
            \node[below left=of f] (e) {$34$};
            \node[above right=of f](g) {$12$};
            \node[below right=of f] (h) {$34$};
            \node[right=of h] (j) {$24$};
            \node[right=of j,yshift=.5cm] (l) {$2$};

      \draw[->] (a) edge (c)
      (b) edge (c)
      (c) edge (d)
      (c) edge (e)
      (d) edge (f)
      (e) edge (f)
      (f) edge (g)
      (f) edge (h)
      (g) edge (j)
      (h) edge (j)
      (j) edge (l);
    \end{tikzpicture}}
  \end{center}
\end{example}


\subparagraph*{Foata normal form} The Foata normal form (FNF) of a trace encodes a maximal parallel execution of the trace \cite{CarFoa69}. To define this notion formally, we use the concept of \emph{steps} from~\cite{Diekert-Muscholl-11}. A \emph{step} is a non-empty subset $S \incl \Sigma$ of pairwise independent letters: it is sometimes called a \emph{clique}, from the point of view of the graph of the independence relation. In a step, all letters can be executed in parallel.  Observe that the set of labels of the minimal elements of a trace $t$ provides a maximal step for $t$. Then, the FNF $\Foata(t)$ of a trace $t$ is a sequence of steps $\varphi = S_1 S_2 \cdots S_m$ where $S_1$ is the set of minimal elements of $t$, and $S_2 \cdots S_m = \Foata(t')$ where $t'$ is the trace obtained by removing all minimal elements from $t$. The FNF is a unique decomposition of the trace into maximal steps. We let $|\varphi|$ be the total number of letters across all the steps of $\varphi$. We write $\Foata(t)_i$ for the $i$-th step of $\Foata(t)$. 

\begin{example}\label{ex:Foata} The FNF of the view of $p_2$ in Example~\ref{ex:views} is 
    \[\Foata(\view_{p_2}(t)) = \{12,34\}\{235\}\{12,34\}\{235\}\{12,34\}\{24\}\{2\}\] 
    Notice that we have depicted traces by putting the events that are in one step of the FNF within the same vertical line.
\end{example}

We denote by $\varphi_i$ the $i$-th step of the FNF $\varphi$. We shall also denote by $\Foatalength t$, the number of steps in the FNF decomposition of the trace $t$, and we call it the \emph{Foata length of $t$}.
In order to simplify further explanations, we suppose that a step $\varphi_i$ exists as the empty set, for all $i$ greater than the length of $\varphi$: we do not depict this infinite sequence of empty sets when we give examples of FNF. 

A crucial feature for our construction is the fact that FNF are increasing with respect to ideals: if a trace $s$ is an
ideal of~$t$, there is an injective correspondence from steps of $s$ to the
first $\Foatalength s$ steps of $t$. 
%
%
Furthermore, if we know the FNF of the views of two (or more) processes $p_1$ and $p_2$, it is possible to deduce the FNF of the view of $\{p_1, p_2\}$: they can accumulate their knowledge simply by taking pairwise unions of the steps. This gives us an algorithm to compute $\view_X(t)$ for some $X \subseteq P$, given $\view_p(t)$ for all $p \in X$. 
In the following, given two FNF $\varphi=S_1\cdots S_m$ and $\varphi'=S'_1\cdots S'_{m'}$, we write $\varphi \cup \varphi'$ for the sequence of steps $(S_1 \cup S'_1) \cdots (S_n \cup S'_n)$ where $n=\max(m,m')$ (remember that we have added empty steps at the end of the FNF so that $S_k$ has a meaning, even for $k$ greater than $\Foatalength \varphi$).

\begin{lemma}[\cite{BerMon26}]
  \label{lem:viewFoata}
  Let $t$ be a trace and $X \subseteq P$. Then, $\Foata(\view_X(t)) = \bigcup\limits_{p\in X} \Foata(\view_p(t))$.
\end{lemma}

\begin{example}
  Continuing Example~\ref{ex:Foata}, we have 
  \[\Foata(\view_{p_1}(t)) = \{12,34\}\{235\}\{12,34\}\{235\}\{12,34\}\{13\}\{35\}\{13\}\] 
  from which we can indeed deduce (by making a stepwise union)
  \[\Foata(\view_{\{p_1, p_2\}}(t)) = \{12,34\}\{235\}\{12,34\}\{235\}\{12,34\}\{13,24\}\{35, 2\}\{13\}\]
\end{example}

\subparagraph*{Regular trace-closed languages and asynchronous automata}
A language $L \subseteq \Sigma^*$ is said to be \emph{trace-closed} (for the independence relation $I$) if for all equivalent $w, w' \in \Sigma^*$, $w \in L$ iff $w' \in L$.
A trace-closed language is \emph{regular} if it is accepted by a deterministic finite state automaton (DFA) $\A = (Q, \Sigma, q_0, \delta, Q_f)$, where $Q$ is the set of states, $q_0$ is the unique initial state, $Q_f$ the accepting states and $\delta\colon Q\times \Sigma \to Q$ the partial function of transitions. We denote by $q\xrightarrow a q'$ if $\delta(q, a)=q'$, and we generalise the notation for a word $w$ by letting $q\xrightarrow{w}q'$ if there is a sequence of transitions $q=q_0\xrightarrow {a_0} q_1 \xrightarrow {a_1} \cdots \xrightarrow{a_{n-1}} q_{n} = q'$ such that $w=a_0 a_1 \ldots a_{n-1}$. As usual, the language recognised by $\A$ is the set of words $w$ such that there exists a run $q_0 \xrightarrow{w}q'$ with $q'\in Q_f$. 
Trace-closed regular languages are recognised by automata having a special syntactic property, called the \emph{diamond property}. 

\begin{definition} For a concurrent alphabet $(\Sigma, I)$, a finite state automaton $\A = (Q, \Sigma, q_0, \delta, Q_f)$ satisfies the \emph{diamond property} if for all $q, q', q'' \in Q$ and $(a, b) \in I$ such that $q \xrightarrow{a} q' \xrightarrow b q''$, there exists $q'''\in Q$ such that $q \xrightarrow{b} q''' \xrightarrow a q''$. 
\end{definition}

For example, the DFA shown in Figure~\ref{fig:fair-example} satisfies the diamond property: $q_0 \xrightarrow{12} q_1 \xrightarrow{34} q_2$ and $q_0 \xrightarrow{34} q_3 \xrightarrow{12} q_2$. Similarly, the DFA shown in Figure~\ref{fig:ccp-automaton} satisfies the diamond property. The diamond property is a sufficient condition for an automaton to recognise a trace-closed language. Indeed, if $\A$ is a DFA satisfying the diamond property then for all $q, q'\in Q$ and equivalent $w, w' \in \Sigma^*$, $q\xrightarrow w q'$ iff $q\xrightarrow{w'} q'$. In particular, we extend this notation to traces $t$ and Foata normal forms $\varphi$, writing $q\xrightarrow t q'$, and $q\xrightarrow{\varphi} q'$. For instance, in the DFA of Figure~\ref{fig:ccp-automaton}, starting from $q_0$, both the words $12~34~13~24$ and $34~12~24~13$ lead to the same state $q_7$.




We will need a function already used in previous work to reconcile the views of several processes that have read independent letters after a point where they last synchronised:
\begin{lemma}[\cite{CorMet93}]\label{lem:diam}
There exists a function $\diam\colon Q \times Q \times Q \times 2^P \to Q$  such that for all states $q_1, q_2, q_3 \in Q$, subsets of processes $X \incl P$, and traces $t_1, t_2, t_3$ with
\begin{enumerate}
  \item $\delta(q_0, t_1) = q_1$, $\delta(q_0, t_1 t_2) = q_2$, $\delta(q_0, t_1 t_3) = q_3$, 
  \item and $\loc(t_2) \incl X$, $\loc(t_3) \incl P \setminus X$,
\end{enumerate}
 we have $\delta(q_0, t_1 t_2 t_3) = \diam(q_1, q_2, q_3, X)$.
\end{lemma}

As an example of the $\diam$ function, notice that for the DFA in Figure~\ref{fig:ccp-automaton}, we have $\diam(q_3, q_6, q_4, \{2, 4\}) = q_7$. Thus, if $t_1 = 12~34$, $t_2 = 24$, $t_3 = 13$, we have $\delta(q_0, t_1 t_2 t_3) = q_7$. An algorithm for computing the $\diam$ function for a DFA satisfying the diamond property is sketched in Appendix~\ref{app:diam}. For the rest of the document, when we consider DFA specifications satisfying the diamond property, we will assume that $\diam$ has been precomputed.

Zielonka's theorem aims at distributing a trace-closed regular language to the individual processes 
by using asynchronous automata.

\begin{definition}
  An asynchronous automaton (AA) over the distributed alphabet $(\Sigma, \loc)$ is a tuple $\B = ((Q_p)_{p \in P}, \Sigma, q^0, (\delta_a)_{a \in \Sigma}, F)$ where
  \begin{itemize}
    \item $Q_p$ is the set of local states of a process $p \in P$,
    \item $\delta_a\colon \Pi_{p \in \loc(a)} Q_p \to \Pi_{p \in \loc(a)} Q_p$ is the transition function associated with $a \in \Sigma$,
    \item $q^0 \in \Pi_{p \in P} Q_p$ is the global initial state, and
    \item $F \subseteq \Pi_{p \in P} Q_p$ is the set of global final states.
  \end{itemize}
  The AA is said to be \emph{finite} if each process has a finite number of states.
  \label{def:asynchronousAutomata}
\end{definition}

 Figure~\ref{fig:DFA'-asynchronous} in Appendix~\ref{app:additional-examples} gives an example of an AA. The semantics of an AA is given by a DFA whose set of states is the product $\Pi_{p\in P} Q_p$. We call these states the \emph{global states} of $\B$. Its initial global state is $q^0$, and its final global states are given by $F$. Moreover, for each letter $a\in \Sigma$, we let $(q_p)_{p\in P} \xrightarrow a (q'_p)_{p\in P}$ if for all $p'\notin\loc(a)$, $q'_p = q_p$, and $\delta_a((q_p)_{p\in \loc(a)}) = (q'_p)_{p\in \loc(a)}$. We can show that this DFA satisfies the diamond property, and thus, we let $\Lt(\B)$ be its trace language, and say that the AA~$\B$ recognises $\Lt(\B)$.

We now state Zielonka's fundamental theorem. An illustration of the theorem is presented in Example~\ref{eg:aa-zielonka} in Appendix~\ref{app:additional-examples}.

\begin{theorem}[\cite{Zie87}]
  For every DFA $\A$ over the concurrent alphabet $(\Sigma, I)$ satisfying the diamond property, and every distributed alphabet $(\Sigma, \loc)$ such that $I_\loc = I$, there exists a finite AA $\B$ over $(\Sigma, \loc)$ such that $\Lt(\A) = \Lt(\B)$.
\end{theorem}

\section{Construction for fair systems}
\label{sec:fair-recall}

We briefly recall the AA construction of \cite{BerMon26} when the DFA we start with is
\emph{fair}: a DFA $\A$ is fair if each process participates at least once in every cycle of $\A$. For example, the DFA of Figure~\ref{fig:fair-example} is fair. We define $d$ to be the least integer such that in every path of length~$d$ all processes necessarily participate: observe that $d$ is at most $1$ plus the length of the longest acyclic path in the DFA. Such a DFA is called  \emph{$d$-fair}. Accordingly, a trace $t$ is said to be \emph{$d$-fair} if for every factor $u$ of any linearisation $w$ of $t$, whenever $|u|\geq d$, we have $\loc(u) = \proc$, the set of all processes. The language of a $d$-fair DFA consists of only $d$-fair traces. 

\subparagraph*{An AA maintaining local views} 
A first na\"ive construction gives an infinite-state AA, which is correct even for unfair DFAs. On a trace $t$, the construction maintains $\view_p(t)$ as the local state of a process $p$. The main idea is to store the view in Foata normal form, i.e., $\Foata(\view_p(t))$. 
The advantage of maintaining views in Foata normal form stems from Lemma~\ref{lem:viewFoata} which explains how to obtain the combined view $\Foata(\view_X(t))$ from the views of each process of $X$. Therefore, when a new letter $a$ appears, all processes in $\loc(a)$ synchronise by taking a stepwise union of their local views, and then adding $a$ as a new step, resulting in their new local state $\Foata(\view_{\loc(a)} (t)) \{a\}$. Clearly, there are infinitely many states for each process. An example of such an AA maintaining views is depicted in Figure~\ref{fig:fair-example}. For fair systems, it turns out that maintaining a bounded suffix of the view suffices.

\subparagraph*{Fair systems: cutting view prefixes} For a $d$-fair trace $t$, when the local view of a process is sufficiently long, it can infer that a prefix of its view is identical to a prefix of the full trace $t$, and moreover, this prefix is present in the view of all the other processes too: 

\begin{lemma}[\cite{BerMon26}]
  \label{lem:everyoneKnowsBeyond2KFoataComponents}
  Let $t$ be a $d$-fair trace and $p \in P$ be a process with a view of length at least $2d-2$. Let $i$ be the largest index such that the steps $\Foata(\view_{p}(t))_i \cdots \Foata(\view_{p}(t))_{\Foatalength{\view_{p}(t)}}$ contain at least $2d-2$ letters. 
  
  Then, for all $p' \in P$, we have
  $\Foata(t)_1 \cdots \Foata(t)_{i-1} = \Foata(\view_{p'}(t))_1 \cdots \Foata(\view_{p'}(t))_{i-1}$.
\end{lemma}

We introduce some terminology: an event $e$ of a trace $t$ is said to be \emph{Second-Order Common Knowledge} ($\sock${}) for process $p$ if $p$ knows that $e$ is present in the current view of every other process; formally, $e$ is present in $\view_{p'}(\view_p(t))$ for all processes $p'$.
Therefore, a process can remove from its view every Foata step in which all events are $\sock${} for it:
Lemma~\ref{lem:everyoneKnowsBeyond2KFoataComponents} ensures that the remaining suffix is bounded linearly w.r.t. $d$.
When processes synchronise, they need to be aware of which of them has cut the most, and then cut their prefixes up to this point, and then finally do the Foata union of the rest of the view. To implement this, local states of a process $p$ in the AA are of the form $(c, q, \varphi)$ where
\begin{itemize}
\item $c \in \N$ is a counter representing the number of letters in the view that have been cut,
\item $q$ is a state of the DFA reached after the prefix of the view that has been cut, and 
\item $\varphi = S_1 S_2 \cdots S_m$ is a suffix of the view that remains, maintained in Foata normal form. 
Moreover, there is at least one letter $a\in S_1$ which is not $\sock$ for $p$: in other words, $p$ does not know whether $a$ is present in the current view of some other process $p'$.
\footnote{This is an optimised version of the algorithm studied in \cite[Section~6]{BerMon26}: the original algorithm proposes to simply cut a prefix of a certain length that ensures storing a short suffix.}
\end{itemize}

To maintain this invariant, a cut operation $\cut(c, q, \varphi)$ is defined on local states, as follows: if $\varphi = S_1S_2 \cdots S_m$, we compute if all the events of $S_1$ are $\sock$ for $p$; if it is the case we compute $\cut(c+|S_1|, q', S_2 \cdots S_m)$ with $q'$ the state such that $q\xrightarrow {S_1} q'$; otherwise we let $\cut(c, q, \varphi) = (c, q, \varphi)$. 
Notice that the only infinite component now is the counter $c$. In fair systems the views of two processes diverge by at most $d$, which can be used to do a modulo counting with $2d$ values for the counter (see \cite[Section~4]{BerMon26}): 

\begin{lemma}[\cite{BerMon26}]
    \label{lem:pairOfProcsAtmostK-1Apart}
    For all $d$-fair traces $t$, and $p, p' \in \proc$, $\big|\length{\view_p(t)} - \length{\view_{p'}(t)}\big| \leq d-1$.
\end{lemma}

\subparagraph*{Fair systems: synchronizing and expanding views} When a new letter appears, the processes participating in this letter reconcile their views by cutting their prefixes up to the process which has cut the most, and then performing a Foata union. More formally, for a letter $a$ with $\loc(a) = X$ and $(c_p, q_p, \varphi_p)$ being the local states of processes $p \in X$, we define $\synchronise((c_p, q_p, \varphi_p)_{p \in X}) = (c', q', \varphi')$  as follows:
\begin{itemize}
\item $c' = \max_{p \in X} (c_p)$; choose a process $p_{\max} \in X$ which maintains the largest counter value $c'$ and let $q' = q_{p_{\max}}$;
\item for each process $p \in X$, if $\varphi_p = S_1 S_2 \dots S_m$, let $j$ be the index such that $S_1 \dots S_{j-1}$ contains exactly $c' - c_p$ letters (existence of such an index is ensured by the construction, and these letters were the ones that were previously cut by $p_{max}$). We let $\varphi'_p$ to be $S_{j} \dots S_m$. Now, define $\varphi' = \bigcup_{p \in X} \varphi'_p$. 
\end{itemize}
Then, we define $\expand((c', q', \varphi'), a) = \varphi' \cdot \{a\}$. Finally, we perform a cut operation on the resulting state: $\cut(c', q', \varphi' \cdot \{a\} )$. This is the final local state of each process $p \in X$ after reading the new letter $a$. In summary, the new state of each $p \in X$ after reading $a$ is:
\begin{align*}
\cut(\expand(\synchronise((c_p, q_p, \varphi_p)_{p \in X}), a))
\end{align*}

\subparagraph*{Fair systems: determining the accepting states} From the local states reached on a trace~$t$, by applying Lemma~\ref{lem:viewFoata}, it is possible to determine the state of $\A$ on reading the global trace $t$. This way, one can determine which combination of local states is accepting.

\begin{theorem}[\cite{BerMon26}]\label{thm:fair-AA}
 Let $\A$ be a $d$-fair DFA over a distributed alphabet $(\Sigma, \loc)$, satisfying the diamond property. There exists a finite AA $\B$ over $(\Sigma, \loc)$ such that $\Lt(\A)=\Lt(\B)$, where every set of local states (for each process) is of size bounded by $O(n \times d\times |\Sigma|^{3d-3})$ with $n$ being the number of states of $\A$.
\end{theorem}


\begin{figure}[tbp]
    \centering
    \scalebox{0.9}{
        \begin{minipage}[c]{0.2\textwidth}
            \centering
            \begin{tikzpicture}[state/.style={rectangle, rounded corners, fill=green!30, inner sep=4pt, draw=gray, thick}, baseline=(current bounding box.center)]
                \begin{scope}[every node/.style={state}]
                    \node (0) at (0,0) {\scriptsize $q_0$};
                    \node (1) at (1.5, 0) {\scriptsize $q_1$};
                    \node (2) at (0, -1.5) {\scriptsize $q_2$};
                    \node (3) at (1.5, -1.5) {\scriptsize $q_3$};
                \end{scope}
                \begin{scope}[->, >=stealth]
                    \draw (-0.7,0) to (0);
                    \draw (0) to node [above] {\scriptsize $12$} (1);
                    \draw (0) to node [left] {\scriptsize $34$} (2);
                    \draw (1) to node [right] {\scriptsize $34$} (3);
                    \draw (2) to node [below] {\scriptsize $12$} (3);
                    \draw (3) to node [right, near end] {\scriptsize $235$} (0);
                \end{scope}
            \end{tikzpicture}
        \end{minipage}%
        \hspace{1.5cm}
        \begin{minipage}[c]{0.7\textwidth}
            \centering
                \renewcommand{\arraystretch}{1.3}
                \arrayrulecolor{gray}
                {\scriptsize 
            \begin{tabular}{l|l|l|l|l|l} 
                &  $12$ & $34$ & $235$ & $34$ & $12$ \\
                \hline
                $p_1$ & $\{\textcolor{blue}{12}\}$ & $\{12\}$ & $\{12\}$ & $\{12\}$ &  $\{12, \blue{34}\} \{\blue{235}\} \{\blue{12}\}$\\
                \hline 
                $p_2$ & $\{\textcolor{blue}{12}\}$ & $\{12\}$ & $\{ 12, \textcolor{blue}{34}\}\{\textcolor{blue}{235}\}$ & $\{12, 34\} \{235\}$ & $\{12, 34\} \{235\} \{\blue{12}\}$ \\
                \hline 
                $p_3$ &  & $\{\textcolor{blue}{34}\}$ & $\{\textcolor{blue}{12}, 34\}\{\textcolor{blue}{235}\}$ & $\{12, 34\} \{235\} \{\blue{34}\}$ & $\{12, 34\} \{235\} \{34\}$\\
                \hline
                $p_4$ &  & $\{\textcolor{blue}{34}\}$ & $\{34\}$ & $\{\blue{12}, 34\} \{\blue{235}\} \{\blue{34}\}$ & $\{12, 34\} \{235\} \{34\} $\\
                \hline 
                $p_5$ &  &  & $\{\textcolor{blue}{12,34}\} \{\textcolor{blue}{235}\}$ & $\{12,34\} \{235\}$ & $\{12, 34\} \{235\}$
            \end{tabular}
                }
        \end{minipage}}
\caption{Construction for fair DFAs}
\label{fig:fair-example}
\end{figure}
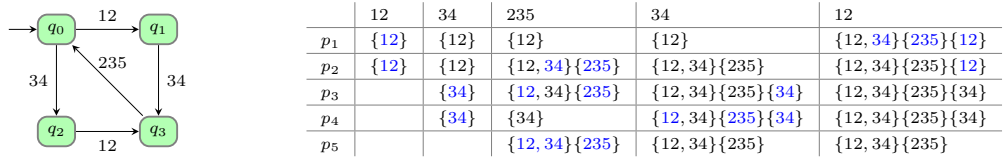

\begin{example}\label{ex:fair-construction}
    We illustrate this construction for the DFA on the left of Figure~\ref{fig:fair-example} with the alphabet and processes described in Example~\ref{ex:distributed-alphabet}. For now, we consider all states to be accepting. This DFA satisfies the diamond property and is fair since all processes participate in all cycles. Moreover, we can set the fairness parameter $d$ to $4$: we cannot set it to $3$ because of the path $q_1 \xrightarrow{34} q_3 \xrightarrow{235} q_0 \xrightarrow{34} q_2$ where process $p_1$ is not participating. On the right of Figure~\ref{fig:fair-example}, we read one by one the letters $12$, $34$, $235$, $34$ and $12$. The counter and state parts of the local states remain $0$ and $q_0$ and are thus not drawn. We put in blue the new information that is added during the synchronisation and expansion. No Foata steps are cut so far: for instance, in the final state reached after reading the letters, the letter $12$ in the first step of the local state of $p_1$, is not $\sock$ for it due to process $p_4$: the view of $p_4$ within the view of $p_1$ is simply $\{34\}$. On the other hand, notice that $p_4$ indeed knows about this event because of the second event $34$ that happens concurrently. We now read letter $235$ (we put back the counter and state parts of the local states):
    \begin{center}\renewcommand{\arraystretch}{1.3}
    \arrayrulecolor{gray}\scriptsize 
    \begin{tabular}{l| l}
    &  $235$ \\
    \hline
    $p_1$ & $(0, q_0, \{12, 34\} \{235\} \{12\})$ \\
    \hline
    $p_2$ &  $(0, q_0, \{12, 34\} \{235\} \{12, \blue{34}\} \{\blue{235}\})~\to~(\blue{3}, \blue{q_0}, \{12, 34\} \{235\})$\\
    \hline
    $p_3$ & $(0, q_0, \{12, 34\} \{235\} \{\blue{12}, 34\} \{\blue{235}\})~\to~ (\blue{3}, \blue{q_0}, \{12, 34\} \{235\})$ \\
    \hline
    $p_4$ &  $(0, q_0, \{12, 34\} \{235\} \{34\})$ \\
    \hline
    $p_5$  & $(0, q_0, \{12, 34\} \{235\} \{\blue{12}, \blue{34}\} \{\blue{235}\})~\to~ (\blue{3}, \blue{q_0}, \{12, 34\} \{235\})$
    \end{tabular}
    \end{center}
    The three processes, once synchronised and after expansion can cut the three first steps since all their events are $\sock$ for them. 
    We thus obtain the cut local states on the right. The counter is incremented by $3$ (the number of cut letters) and the new state is the one reached in $\A$ from $q_0$ by reading $12$, $34$, and $235$.
\end{example}


\section{Connectedly Communicating Processes}
\label{sec:ccp-basic}

We now extend the algorithm of Section~\ref{sec:fair-recall}, 
to allow more DFAs. Consider the DFA in Figure~\ref{fig:ccp-automaton}. This DFA is not fair: in the cycle $q_4 \xra{35} q_5 \xra{13} q_4$, processes $2$ and $4$ do not participate. However, in all the paths starting from $q_4$ or $q_5$, processes $2$ and $4$ do not communicate with $1, 3, 5$, neither directly nor indirectly through a sequence of intermediate synchronisations. Such systems have been called \emph{connectedly communicating processes} in~\cite{MadThi05}. 
This notion was originally defined on an asynchronous automaton. Here, we adapt it to the case of DFAs.  


Processes $p$ and $p'$ are said to be \emph{separated} in a trace $t$ over $(\Sigma, \loc)$ if there exists a linearisation $uu'$ of $t$ such that (1) $\loc(u) \cap\loc(u')=\emptyset$ ($u$ and $u'$ are independent) and (2) $|u|_{p'} =  |u'|_p = 0$ ($p'$ does not occur in $u$, and $p$ does not occur in $u'$). If such a linearisation does not exist for $t$, we say that $p$ and $p'$ are \emph{potentially communicating} in $t$. For a word $w \in \Sigma^*$, we say $p$ and $p'$ are separated in $w$ if they are separated in the trace corresponding to $w$. This definition of separation allows one to semantically define the notion of languages over connectedly communicating processes: intuitively, two processes that are potentially communicating (i.e., not separated yet) must learn about each other often enough. 

\begin{definition}
  A trace language $L$ over a distributed alphabet is \emph{connectedly communicating} if there exists $k \ge 0$ such that for all traces $t \in L$, all processes $p, p'$ and linearisations $uvw$ of $t$, if $|v|_p \ge k$ and $|v|_{p'} = 0$, then $p$ and $p'$ are separated in $w$.
  \label{def:k-communicating-langs}
\end{definition}
In~\cite{MadThi05}, a system is called a $k$-CCP if it satisfies the above definition with the value taken to be $k$. In our work, we consider an alternate characterization and a different parameter which we will describe next. 



\subparagraph*{Testing connected communication} We formalize the potential communication between processes as an equivalence relation at each state of the DFA. This relation turns out to be identical for all states within a strongly connected component (SCC): an SCC is a subset of states such that there is a path between every pair of states, and moreover it is maximal in the sense that no strict superset of states satisfies this property. This gives a notion of communicating partners for a process at each SCC. 


\begin{definition}\label{def:potential-communication-relation}
 Let $\A$ be a DFA and let $s$ be a state in $\A$. We define an equivalence relation $\equiv_s$ over the set of processes, called the \emph{potential communication} relation at $s$, as the reflexive transitive closure of the relation $\equiv'_s$ defined as follows: $p \equiv'_s p'$ if there exists a path labelled $w$ in $\A$ starting from $s$ such that $p$ and $p'$ are potentially communicating. We write $[p]_s$ for the equivalence class of process $p$ at state $s$ w.r.t.~$\equiv_s$.

 We define \emph{separation depth} of $\A$ to be the maximum number of times a refinement of the potential communication relation can happen in a path of $\A$.
\end{definition}

Notice that if $s \xra{a} s'$ is a transition of $\A$, then $\equiv_{s'}$ is a refinement of $\equiv_s$: if $p$ and $p'$ are potentially communicating in $s'$, then they potentially communicate from $s$ as well. This observation helps to prove the following lemma about states within the same SCC: 

\begin{restatable}{lemma}{PCSameInSCC}
    \label{lem:potential-communication-same-in-SCC}
    Let $s, s'$ be two states which are in the same SCC of $\A$. Then the equivalence relations $\equiv_s$ and $\equiv_{s'}$ coincide.
\end{restatable}

Thanks to the above lemma, for each SCC $\Gamma$ of $\A$, we can define $\equiv_\Gamma$ to be equal to $\equiv_s$ for some $s \in \Gamma$. The relation $\equiv_\Gamma$ will be called the potential communication relation for $\Gamma$. Figure~\ref{fig:ccp-automaton} shows this relation for the depicted DFA. Let us reconsider the separation depth of Definition~\ref{def:potential-communication-relation}. Since the potential communication relation is the same within an SCC, the number of times it can get refined is at most the height of a topological ordering of the SCCs (this is an ordering where $\Gamma$ is less than $\Gamma'$ if there is a path from a state in $\Gamma$ to a state in $\Gamma'$). It can be smaller, since the relation may stay the same in several SCCs. If $\A$ is fair, then the separation depth is equal to $0$. Moreover, we can compute these equivalence relations in polynomial time by visiting the SCCs in reverse topological order. 

The next lemma provides an alternate characterization for a DFA to be connectedly communicating, by making use of the potential communication relation. This characterization allows one to test whether a given DFA is connectedly communicating.

\begin{restatable}{lemma}{CCfair}
    \label{lem:CC-fair}
    A trim DFA $\A$ satisfying the diamond property is connectedly communicating if and only if for every SCC $\Gamma$, and every equivalence class $X\subseteq P$ of $\equiv_{\Gamma}$, the SCC $\Gamma$ is fair with respect to $X$, i.e.~whenever some process of $X$ participates in a cycle of $\Gamma$, all processes of $X$ participate in it as well. 
\end{restatable}

We are now in a position to define a counterpart of the fairness parameter in the context of connectedly communicating DFAs. Consider a DFA which is connectedly communicating. For an SCC $\Gamma$ of $\A$, and an equivalence class $X$ of $\equiv_\Gamma$, the fairness parameter $d_{\Gamma, X}$ is the minimum number such that in every path $\sigma$  of length $\ge d_{\Gamma, X}$ inside $\Gamma$ that consists only of $X$ actions (i.e.~actions of $\bigcup_{p\in X} \Sigma_p$), we have $\loc(\sigma) = X$. 
The \emph{maximum fairness parameter}~$d$ is then defined to be the maximum over $d_{\Gamma, X}$ for all $\Gamma$ and all equivalence classes $X$ of $\Gamma$.

\begin{figure}
\centering
\scalebox{0.9}{
\begin{minipage}[c]{.45\textwidth}
    \centering
\begin{tikzpicture}[greenstate/.style={rectangle, rounded corners, fill=green!30, inner sep=4pt, draw=gray, thick}, redstate/.style={rectangle, rounded corners, fill=red!30, inner sep=4pt, draw=gray, thick}, graystate/.style={rectangle, rounded corners, fill=gray!30, inner sep=4pt, draw=gray, thick}, baseline=(current bounding box.center), yellowstate/.style={rectangle, rounded corners, fill=yellow!30, inner sep=4pt, draw=gray, thick}, baseline=(current bounding box.center)]
\begin{scope}[every node/.style={greenstate}]
    \node (0) at (0,0) {\scriptsize $q_0$};
    \node (1) at (1.5, 0) {\scriptsize $q_1$};  
    \node (2) at (0, -1.5) {\scriptsize $q_2$};
    \node (3) at (1.5, -1.5) {\scriptsize $q_3$};       
\end{scope}
\begin{scope}[every node/.style={redstate}]
 \node (4) at (3, -1.5) {\scriptsize $q_4$};
\node (5) at (4.5, -1.5) {\scriptsize $q_5$};
\end{scope}
\begin{scope}[every node/.style={graystate}]
\node (6) at (1.5, -3 ) {\scriptsize $q_6$};
    \node[double] (7) at (3, -3) {\scriptsize $q_7$};
    \node (8) at (4.5, -3) {\scriptsize $q_8$};
\end{scope}
\begin{scope}[->, >=stealth]
     \draw (-0.7,0) to (0);
    \draw (0) to node [above] {\scriptsize $12$} (1);
    \draw (0) to node [left] {\scriptsize $34$} (2);
    \draw (1) to node [right] {\scriptsize $34$} (3);
    \draw (2) to node [below] {\scriptsize $12$} (3);
    \draw (3) to node [right, near end] {\scriptsize $235$} (0);
    \draw (3) to node [above] {\scriptsize $13$} (4);
    \draw (3) to node [left] {\scriptsize $24$} (6);
    \draw (4) to [bend left=20] node [above] {\scriptsize $35$} (5);
    \draw (5) to [bend left=20] node [below] {\scriptsize $13$} (4);
    \draw (4) to node [left] {\scriptsize $24$} (7);
    \draw (5) to node [right] {\scriptsize $24$} (8);
    \draw (6) to node [above] {\scriptsize $13$} (7);
    \draw (7) to [bend left=20] node [above] {\scriptsize $35$} (8);
    \draw (8) to [bend left=20] node [below] {\scriptsize $13$} (7);
    \draw (6) to [loop below] node [below] {\scriptsize $2,4$} (6);
    \draw (7) to [loop below] node [below] {\scriptsize $2,4$} (7);
    \draw (8) to [loop below] node [below] {\scriptsize $2,4$} (8);
\end{scope}

\draw[rounded corners,dotted] (-0.5,0.5) rectangle (2,-2);
\node() at (2.3,0) {$\Gamma_1$};
\draw[rounded corners,dotted] (2.5,-.95) rectangle (5,-2.05);
\node() at (4.25,-.75) {$\Gamma_2$};
\draw[rounded corners,dotted] (1,-2.45) rectangle (2,-4);
\node() at (.7,-3.4) {$\Gamma_3$};
\draw[rounded corners,dotted] (2.5,-2.45) rectangle (5,-4);
\node() at (5.3,-3.4) {$\Gamma_4$};
\end{tikzpicture}
\end{minipage}
\hspace{1cm} 
\begin{minipage}[c]{.4\textwidth}
    \begin{align*}
       &p_1 \equiv_{\Gamma_1} p_2 \equiv_{\Gamma_1} p_3 \equiv_{\Gamma_1} p_4 \equiv_{\Gamma_1} p_5 \\
       &p_1 \equiv_{\Gamma_2} p_3 \equiv_{\Gamma_2} p_5 \qquad p_2 \equiv_{\Gamma_2} p_4 \\ 
       &p_1 \equiv_{\Gamma_3} p_3 \equiv_{\Gamma_3} p_5 \\
       &p_1 \equiv_{\Gamma_4} p_3 \equiv_{\Gamma_4} p_5 
    \end{align*}
\end{minipage}}
\caption{On the left, a connectedly communicating DFA with four SCCs. On the right, the potential communication relation in each SCC (in $\Gamma_3$ and $\Gamma_4$, processes $2$ and $4$ are separated from the other processes).}
\label{fig:ccp-automaton}
\end{figure}
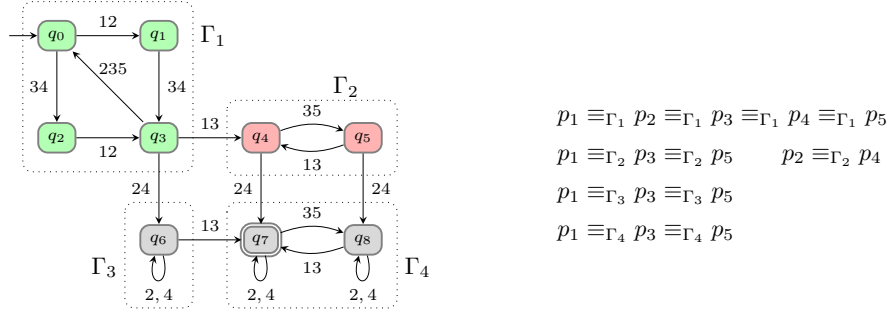

\begin{example}\label{ex:CC-DFA}
    Consider the DFA $\A'$ on the left of Figure~\ref{fig:ccp-automaton}, with $q_7$ being accepting, and working over the distributed alphabet of Example~\ref{ex:distributed-alphabet}. There are four SCCs depicted with dotted rectangles: for instance, 
    $\Gamma_1 = \{q_0, q_1, q_2, q_3\}$. 
    We have coloured the states according to the potential communication relations that are described on the right of Figure~\ref{fig:ccp-automaton}: the states of two SCCs have the same colours when they share the same potential communication relation. This is the case for $\Gamma_3$ and $\Gamma_4$. Even if processes $p_1$, $p_3$, and $p_5$ are not communicating inside~$\Gamma_3$, they are still potentially communicating because of $\Gamma_4$. 
    Notice that $\A'$ is not fair: for instance, in the reachable loop around $q_6$, with letter $2$, only process $p_2$ is participating. 
    
    The separation depth of $\A$ is $2$, since there exists a path going through $q_3$, $q_4$ and $q_7$, where the potential relation communication has been refined twice. The maximum fairness parameter is $4$ since the fairness parameter is $4$ in SCC $\Gamma_1$ (same reason as in Example~\ref{ex:fair-construction}), $2$ for the equivalence class $\{p_1, p_3, p_5\}$ in $\Gamma_2$ and $\Gamma_4$, and $1$ for the others.

\end{example}

\section{Construction for connectedly communicating specifications}
\label{sec:ccp-construction}
We are now ready to formally state our main contribution: 

\begin{theorem}\label{thm:CC-AA}
  For every connectedly communicating DFA $\A$ over a distributed alphabet $(\Sigma, \loc)$ satisfying the diamond property, 
  there exists a finite AA $\B$ over $(\Sigma, \loc)$ such that $\Lt(\A)=\Lt(\B)$, where every set of local states (for each process) is of size bounded by $O\big((n \times d\times |\Sigma|^{3d-3})^{s+1}\big)$ with $n$ the number of states of $\A$, $d$ the maximum fairness parameter of $\A$, and $s$ the separation depth of $\A$.
\end{theorem}
A lower bound on the size of an AA equivalent to $\A$ is provided in~\cite{BerMon26} in the case of a fair system, i.e.~in which the separation depth is $s=0$: this shows that the $d$ component in the exponent is unavoidable. We first provide an overview of our construction. 

\subparagraph*{Overview} As in the case of fair systems, each process
maintains its local view of the trace read so far, in Foata normal form. For
fair systems, all Foata steps which are $\sock${} for the process could be
deleted. This is no longer possible. Consider the trace of
Example~\ref{ex:trace}. Here is the view of process $p_1$ in Foata normal form
(ignore the dashed line in between for now):
\begin{align*}
\left[ \begin{tabular}{c} 12 \\ 34\end{tabular} \right] 
\left[ \begin{tabular}{c} 235 \\ \hspace{1mm} \end{tabular} \right] 
\left[ \begin{tabular}{c} 12 \\ 34\end{tabular} \right] 
\left[ \begin{tabular}{c} 235 \\ \hspace{1mm} \end{tabular} \right] 
\left[ \begin{tabular}{c} 12 \\ 34\end{tabular} \right] 
~\dvline~
\left[ \begin{tabular}{c} 13 \\ \hspace{1mm} \end{tabular} \right] 
\left[ \begin{tabular}{c} 35 \\ \hspace{1mm} \end{tabular} \right]
\left[ \begin{tabular}{c} 13 \\ \hspace{1mm} \end{tabular} \right]
\end{align*}
Notice that the events in the first four Foata steps are $\sock${} for $p_1$.  The fifth step $\{12, 34\}$ (just before the dashed line) is not $\sock${} since $p_1$ does not know whether $p_2$ is aware of event $34$. Similarly, the events $13~35~13$ are not $\sock${}. More importantly, $13~35~13$ can never be $\sock${} in any extension of the view consistent with $\A'$. This is because in the state $q_4$ reached after this trace, $p_2$ and $p_4$ have already separated from $p_1$ (see right of Figure~\ref{fig:ccp-automaton}) and hence the events where they participate would no longer appear in the view of $p_1$. Therefore, the construction cutting $\sock${} events does not work anymore. 

A natural thought is to focus simply on the potentially communicating partners of $p_1$ at the end of the view, namely $[p_1]_{q_4} = \{p_1, p_3, p_5\}$.  In
the current view, $p_1$ can deduce that all Foata steps, except for the last $\{13\}$ are present in the views of $p_3$ and $p_5$ too: we call these events to be $\sock[\{p_1, p_3, p_5\}]$ -- $\sock${} restricted to $\{p_1, p_3, p_5\}$. In any extension, $p_1$ will synchronise only with $p_3$ or $p_5$. Does this mean we can cut this entire prefix? Unfortunately, not. Even though in any extension, $p_1$ can synchronise only with $p_3$ or $p_5$, there is a final step to determine whether the global state reached on a trace is accepting or not. In the construction of Section~\ref{sec:fair-recall}, the reconciliation happens by computing the joint view of all the processes, taking one at a time: $\view_{\{p_1,p_2\}}$ is computed from $\view_{p_1}$ and $\view_{p_2}$ by cutting appropriate number of steps and then taking Foata union; then $\view_{\{p_1,p_2,p_3\}}$ is computed from $\view_{\{p_1, p_2\}}$ and $\view_{p_3}$ and so on. Therefore, in the example above, the view of $p_1$ would be synchronised with views of $p_2$ and $p_4$ as well, in the final step, to determine acceptance.

This brings us to the question: \emph{when can we safely cut an event?} Here is our proposal. Assume the view is a sequence of Foata steps $S_1 \dots S_m$, and the run of the DFA is $q_0 \xra{S_1} q_1 \xra{S_2} \cdots \xra{S_m} q_m$. We show that $S_i$ can be safely cut as long as every event in prefix $S_1 \dots S_i$ lies in $\sock[[p]_{q_i}]$, the second-order common knowledge of $p$ restricted to partners of $p$ at \emph{the next state $q_i$} (and not the state $q_m$ reached at the end of the view). For the above example, the state reached before the dashed line is $q_3$, after which we have $q_3 \xra{\{13\}} q_4 \xra{\{35\}} q_5 \xra{\{13\}} q_4$. At state $q_3$, partners of $p_1$ are all processes, i.e., $[p_1]_{q_3} = \{p_1, \dots, p_5\}$. Observe that  $34$ is not $\sock[[p_1]_{q_3}]$. Therefore, we do not cut the last $\{12, 34\}$ step. But, we can cut the following $13~35$ since all steps up to $13$ are $\sock[[p_1]_{q_4}]$ and all steps up to $35$ are $\sock[[p_1]_{q_5}]$. The idea is that when we cut a Foata step in the view of a process, it knows that all events in its view up to this step are visible to each of its communicating partners.

In this mechanism to cut steps, we may encounter cuts in the middle of the view, unlike in the fair case where the cuts happened only at a
prefix of the view. To book-keep the cuts, we mark the places where the cuts start (denoted by the dashed line).
This amounts to splitting the view into \emph{phases}. For the view above, the
first phase is the part to the left of the dashed line, and the second phase is
to its right. Within each phase, a prefix of Foata steps can be cut. In the first phase above, the prefix consisting of the first four Foata steps can be cut by applying our rule. 
For each phase, we will maintain a
triple $(c, q, \varphi)$ as in the fair construction of
Section~\ref{sec:fair-recall}. For the above view of $p_1$, the state reached
would be $(6, q_0, \{12, 34\})(2, q_5, \{13\})$. 

Once we have local states in this form, the next challenge is to be able to
compute the joint states when multiple processes synchronise. For instance,
assume a letter $a$ is being read, with $\loc(a) = \{p, p'\}$. We need to
reconcile the current views of $p$ and $p'$. Due to our cutting strategy, any new information that $p$ can learn from $p'$ would be after the last cut made in $p$. However there could be an event which is cut in one process, but not cut yet in the other. 
The key observation is that $p$ and $p'$ can
be at most one phase apart, and only the last phase of the lagging process needs
to be synchronised with the last two phases of the leading process. All these
ideas are concisely captured in Invariants (A), (B), (C) and (D) of
Definition~\ref{def:invariant}.

The final difficulty is to determine which global states are accepting. To do so, we need to be able to determine the state of $\A'$ reached at the end of the global trace, simply by looking at the above maintained abstraction of the local views. Here is where we make use of Invariant (E) (Definition~\ref{def:invariant}) and the $\diam$ function (Section~\ref{sec:prelim}). We explain the determination of final states later in the section, after formalizing the local states and transitions of the AA $\B$. 

\subparagraph*{States of $\B$} 
We start by explaining how to produce an AA with infinitely many local states, due to the counters that will be unbounded. Later, we explain how to make this part finite by using the same modulo counting technique as in~\cite{BerMon26}. The local states of a process $p$ in the (infinite) AA we are building have the shape:
$\sigma_p := (c_{p,i}, q_{p,i}, \varphi_{p,i})_{1 \leq i \leq \ell_p}$, 
where each triple is called a \emph{phase}, that contains a counter value $c_{p,i} \in \N$ to track the number of cut letters in this phase, a state $q_{p,i}$ of $\A$, and the Foata normal form $\varphi_{p,i}$ of a trace. In the fair case discussed in Section~\ref{sec:fair-recall}, only a single phase is maintained along the computation. 

The local initial state of a process $p$ is given by $(0, q_0, \varepsilon)$, with $q_0$ being the initial state of $\A$.
Along our construction, we will maintain an invariant on the global states we reach in the AA, that is trivially fulfilled for the empty trace.

\begin{definition}\label{def:invariant}
A trace $t$ \emph{satisfies the invariant} if by reading it in $\B$, we reach a global state $\big((c_{p,i}, q_{p,i}, \varphi_{p,i})_{1\leq i\leq \ell_p}\big)_{p \in P}$ s.t.
for all processes $p$, we have states $s_{p,0}, s_{p,1}, \ldots, s_{p,\ell_p}$ with $s_{p,0}$ the initial state of $\A$ and $q_{p,i} \xrightarrow{\varphi_{p,i}} s_{p,i}$ for all~$i$, 
and the following conditions hold: 
\begin{enumerate}
    \item[$\mathrm{(A)}$] for every process $p$, for all $2\leq i \leq \ell_p$, $c_{p,i} >0$: \emph{all phases, except possibly the first, have experienced a cut};
    \item[$\mathrm{(B)}$] for every process $p$, for all $1\leq i< \ell_p$, we have $[p]_{s_{p,i+1}} \subsetneq [p]_{s_{p,i}}$: \emph{in-between two phases, the potential communication relation is strictly refined};
    \item[$\mathrm{(C)}$] for every process $p$, $\Foata(\view_p(t))$ can be decomposed as $\psi_{p,1} \varphi_{p,1} \psi_{p,2} \varphi_{p,2} \cdots \psi_{p,\ell_p} \varphi_{p,\ell_p}$ with $|\psi_{p,i}| = c_{p,i}$ and $s_{p,i-1} \xrightarrow{\psi_{p,i}}q_{p,i}$ for all $1\leq i\leq \ell_p$: \emph{the counters exactly count the number of cut letters in the view of a process; in particular, the state of $\A$ reached after reading $\view_p(t)$ is equal to $s_{p,\ell_p}$};
    \item[$\mathrm{(D)}$] for each pair of processes $p, p'$ such that $p'\in [p]_{s_{p,\ell_p}}$ and $p \in [p']_{s_{p',\ell_{p'}}}$, we have $|\ell_p-\ell_{p'}|\leq 1$, and if $\ell_{p'} =\ell_{p}-1$ then $c_{p',\ell_{p'}} \leq c_{p,\ell_{p'}}$: \emph{two processes that are still potentially communicating can be at most one phase apart, and if one of the processes has one less phase, it has cut no more letters than the other process};
    \item[$\mathrm{(E)}$] for each pair of processes $p, p'$, letting $\ell_{p, p'}$ be the last phase $i$ such that $p'\in [p]_{s_{p,i}}$, and similarly for $\ell_{p',p}$, we have for all $1\leq j\leq \min(\ell_{p,p'},\ell_{p',p})-1$, 
    $(c_{p,j}, q_{p,j}, \varphi_{p,j}) = (c_{p',j}, q_{p',j}, \varphi_{p',j})$: \emph{the last time two processes were still potentially communicating, the common past phases, except for the last one, are equal}. 
\end{enumerate}
\end{definition}

\subparagraph*{Cut and split operation}

One of the more subtle operations that we have to perform along the computation is the splitting into phases in order to be able to perform more cuts. To do so, we introduce an operation $\dvdcut_{p,q}$ with $p$ a process and $q$ a state of $\A$. This operation takes the Foata normal form $\varphi$ of a trace $t$ 
as parameter, and returns a sequence of phases. Moreover, for a set of processes $X$, we say an event $e$ is \emph{second-order common knowledge among $X$} or $\sock[X]$ for a process~$p$ 
if for all processes~$p' \in X$, $\view_{p'}(\view_p(t))$ contains $e$, i.e.~$p$ knows that every other process in $X$ knows the event~$e$. 

Let $\varphi = S_1 S_2 \cdots S_m$ be the Foata steps, and let 
$q=q_0 \xra{S_1} q_1 \xra{S_2} q_2 \cdots \xra{S_m} q_m$ be the run reading $\varphi$ from $q$. Here, $\varphi$ is part of the view of process $p$. 
If for all $i$, every event in $S_1 \dots S_i$ is $\sock[[p]_{q_i}]$ for $p$, then every letter can safely be cut and we let
    $\dvdcut_{p,q}(\varphi) = (|\varphi|, q_m, \varepsilon)$.
Otherwise, let $i$ be the minimal index such that some letter of $S_1 \dots S_i$ is not $\sock[[p]_{q_i}]$ for $p$. Due to minimality of $i$, the steps $S_1\cdots S_{i-1}$ can safely be cut, but not~$S_i$. There are two further cases depending on whether we should split the last phase or not: 
    \begin{itemize}
        \item If for all $j>i$, at least a letter of $S_i, S_{i+1}, \dots, S_j$ is not $\sock[[p]_{q_j}]$ for $p$, nothing else can be cut, and we thus let 
        $\dvdcut_{p,q}(\varphi) = (|S_1\cdots S_{i-1}|, q_{i-1}, S_i\cdots S_m)$.
        \item Otherwise, there exists a minimal $j>i$ such that all events in $S_i, S_{i+1}, \dots, S_j$ are $\sock[[p]_{q_j}]$. 
        By minimality, $S_{j-1}$ cannot be cut: there must exist a process in $[p]_{q_{j-1}}$ that, according to $p$, does not know about a letter of $S_i, S_{i+1}, \dots, S_{j-1}$, and this process cannot be among $[p]_{q_j}$, thus we have $[p]_{q_j} \subsetneq [p]_{q_{j-1}}$ (which will help us get invariant (B)). We split to build a first phase 
        $(|S_1\cdots S_{i-1}|, q_{i-1}, S_i\cdots S_{j-1})$. We then continue this procedure with the steps $S_j\cdots S_m$, and we thus define inductively \[\dvdcut_{p,q}(\varphi) = (|S_1\cdots S_{i-1}|, q_{i-1}, S_i\cdots S_{j-1}) \; \dvdcut_{p,q_{j-1}}(S_j\cdots S_m)\]
    \end{itemize}

\subparagraph*{\texorpdfstring{Transitions of $\B$}{Transitions of B}}

Before reading a letter $a$, all processes in $\loc(a)$ must synchronise their local states $(\sigma_p)_{p\in \loc(a)}$. We let $\synchronise\big((\sigma_p)_{p\in \loc(a)}\big)$ be the local state obtained after synchronisation, as follows. We describe only this synchronisation between two processes $p$ and $p'$ participating in $a$. The general case follows by induction. During our explanations, we use the invariant of Definition~\ref{def:invariant} that we want to maintain.

Let $\sigma_p = (c_{i}, q_{i}, \varphi_{i})_{1\leq i\leq \ell}$ be the local state of process $p$, and $\sigma_{p'} = (c'_{i}, q'_{i}, \varphi'_{i})_{1\leq i\leq \ell'}$ be the local state of process $p'$. Without loss of generality, we suppose that $\ell'\leq \ell$ and if $\ell = \ell'$, then $c_{\ell'} \leq c_{\ell}$. Since $p$ and $p'$ are potentially communicating, by invariant~(D), $\ell'\in\{\ell,\ell-1\}$. By invariant~(E), we know that for all $1 \le i \le \ell'-1$, $(c_i, q_i, \varphi_i) = (c'_i, q'_i, \varphi'_i)$. We keep those phases unchanged during the synchronisation.

If $\ell' = \ell-1$, we make $p'$ catch up as follows. By invariant~(D), we know that $c'_{\ell'} \leq c_{\ell'}$: process $p$ has cut more than $p'$ in the $\ell'$-th phase. We thus cut the first $c_{\ell'}-c'_{\ell'}$ letters from~$\varphi'_{\ell'}$ (as in~\cite{BerMon26}, these letters are ensured to represent whole steps in the Foata normal form). The remaining Foata steps are then of the form $\varphi_{\ell'}\psi$. In the state after synchronisation, we add the phase $(c_{\ell'}, q_{\ell'}, \varphi_{\ell'})$. 
We temporarily add a phase $(0, q_{\ell'}, \psi)$ to process $p'$.

For the last phase (either $p'$ had a phase of index $\ell'=\ell$, or we have created one just above), we must synchronise the view of $p$ that is $(c_\ell, q_\ell, \varphi_\ell)$, and the view of $p'$ that is~$(c'_\ell,q'_\ell, \varphi'_\ell)$. We still have $c'_\ell \leq c_\ell$ (either it is by hypothesis, or we have $c'_\ell=0$ by the case just above). We can thus still cut $c_\ell-c'_\ell$ letters in the phase of $p'$, to obtain some Foata steps $\psi'$. Guided by Lemma~\ref{lem:viewFoata}, we then perform a componentwise union of the steps of $\varphi_\ell$ and $\psi'$ to obtain Foata steps $\psi''$. We thus obtain 
$\synchronise\big(\sigma_p, \sigma_{p'}\big) = (c_1, q_1, \varphi_1)\ldots (c_{\ell-1}, q_{\ell-1}, \varphi_{\ell-1})(c_{\ell}, q_\ell, \psi'')$.

Once all participating processes of $\loc(a)$ have synchronised, the updated last phase $(c_\ell, q_\ell, \psi'')$ can then be expanded with a new step $\{a\}$. To finish, we apply a phase split and replace the last phase $(c_{\ell}, q_\ell, \psi'')$ by the phases $\dvdcut_{p, q_\ell}(\psi''\{a\})$.
We can show that if a trace $t$ satisfies the invariant of Definition~\ref{def:invariant}, then the trace $ta$, obtained by adding a letter~$a$ also satisfies this invariant, which proves the correctness of the construction. A detailed example illustrating an entire run for each process is provided in Example~\ref{ex:CC-construction} in Appendix~\ref{app:additional-examples}.

\subparagraph*{\texorpdfstring{Determining the accepting states of $\B$}{Determining the accepting states of B}}

At this point, the invariant of Definition~\ref{def:invariant} ensures that after reading a trace $t$, each process $p$ reaches a local state from which the state reached in $\A$ on reading $\view_p(t)$ can be computed. However, we need to be able to compute the state reached in $\A$ after reading the whole trace $t$, in order to declare if the global state of the AA is accepting or not. Our construction proceeds in two steps: (1) compute a tree depicting how two processes separate at each phase, (2) a bottom-up computation over this tree computing the states reached by $\view_X(t)$ for a  subset $X$ of processes at each stage, until $X$ equals the whole set of processes, which will then give us the required state. We first exemplify our construction, before formalising it.

\begin{example}\label{ex:acc-intersection}
Consider the DFA of Figure~\ref{fig:CC-DFA-intersection} over the set of processes $\{p_1, p_2, p_3, p_4\}$ and a new alphabet following the same conventions as before (in particular, $\loc(134) = \{p_1, p_3, p_4\}$). It satisfies the diamond property, is trim and connectedly communicating. The potential communication relations are depicted on the right for each SCC.
We read the trace $12 \; 34\; 13\; 24\; 13$ that is accepted by the DFA:
\begin{center}\renewcommand{\arraystretch}{1.3}%
\arrayrulecolor{gray}%
\scriptsize%
\scalebox{0.95}{
\begin{tabular}{c | l | l | l | l | l}
    & $12$ & $34$ & $13$ & $24$ & $13$ \\
    \hline
    $p_1$ & $(0,q_0,\{\blue{12}\})$ & $(0,q_0,\{12\})$ & $(0,q_0,\{12,\blue{34}\}\{\blue{13}\})$ & $(0,q_0,\{12,34\}\{13\})$ &
    $(0,q_0,\{12,34\}\{13\})\blue{(1,q_5,\varepsilon)}$\\\hline
    $p_2$ & $(0,q_0,\{\blue{12}\})$ & $(0,q_0,\{12\})$ & $(0,q_0,\{12\})$ & $(0,q_0,\{12,\blue{34}\})\blue{(1,q_6,\varepsilon)}$ & $(0,q_0,\{12,34\})(1,q_6,\varepsilon)$ \\\hline
    $p_3$ & $(0,q_0,\varepsilon)$ & $(0,q_0,\{\blue{34}\})$ & $(0,q_0,\{\blue{12},34\}\{\blue{13}\})$ & $(0,q_0,\{12,34\}\{13\})$ & $(0,q_0,\{12,34\}\{13\})\blue{(1,q_5,\varepsilon)}$\\\hline
    $p_4$ & $(0,q_0,\varepsilon)$ & $(0,q_0,\{\blue{34}\})$ & $(0,q_0,\{34\})$ & $(0,q_0,\{\blue{12},34\})\blue{(1,q_6,\varepsilon)}$ &  $(0,q_0,\{12,34\})(1,q_6,\varepsilon)$
\end{tabular}}
\end{center}
In the final global state computed, each process has two phases. Processes $p_1$ and $p_3$ reach states $q_4$ after Phase $1$, and $q_5$ after Phase $2$. Processes $p_2$ and $p_4$ reach $q_3$ after Phase $1$ and $q_6$ after Phase $2$. From these states, we can infer the following tree:

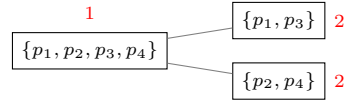
\begin{center} 
\begin{tikzpicture}
 \node [draw] (0) at (0,0) {\scriptsize $\{p_1, p_2, p_3, p_4 \}$};
 \node [draw] (1) at (2.5, 0.4) {\scriptsize $\{p_1, p_3\}$};
 \node [draw] (2) at (2.5, -0.4) {\scriptsize $\{p_2, p_4\}$};
 \node [red] at (0, 0.5) {\scriptsize $1$};
 \node [red] at (3.3, 0.4) {\scriptsize $2$};
 \node [red] at (3.3, -0.4) {\scriptsize $2$};
 \draw [thin, gray] (0) to (1);
 \draw [thin, gray] (0) to (2);
\end{tikzpicture}
\end{center}
The tree reveals that at the end of Phase $1$, all processes are still communicating, whereas by the end of Phase $2$, the communication relation has split to $\{p_1, p_3\}$ and $\{p_2, p_4\}$. We now move on to the second step. First, we resynchronise processes that are still communicating, i.e.~$p_1$ with $p_3$, and $p_2$ with $p_4$, but nothing changes since their respective last seen letter was a synchronisation with their partner. We deduce that $\view_{\{p_1, p_3\}}(t) = q_5$ and $\view_{\{p_2, p_4\}} = q_6$. Now, we move on to the root node of the above tree, and look at the entry in Phase $1$ for all the processes.  Then, to reconcile the equivalence classes $\{p_1, p_3\}$ and $\{p_2, p_4\}$, we take an intersection of Foata normal forms in the last phase they have in common: by definition, we know that the symmetric difference in-between the Foata normal forms will be actions seen by an equivalence class, but not by the other ones and vice versa, thus independent actions. Here, the last letter $13$ of the first phase is removed by the intersection (since it is not known by $p_2$ and $p_4$), to obtain $(0,q_0,\{12,34\})(1,q_5,\varepsilon)$ for $\{p_1, p_3\}$, and $(0,q_0,\{12,34\})(1,q_6,\varepsilon)$ for $\{p_2, p_4\}$. 
What we gain after doing this computation is that all processes now share the state reached in-between the two phases, that is state $q_3$. To obtain the state of $\A$ reached at the end of the run on the trace, we compute $\diam(q_3, q_5, q_6, \{p_1, p_3\}) = q_8$. We can then declare the global state to be accepting.
\end{example}

We present some of the technical details here, and the rest in \Cref{app:proofs-ccp-construction}. For each process~$p$, let $\sigma_p = (c_{p,i}, q_{p,i}, \varphi_{p,i})_{1 \le
i \le \ell_p}$ be the local state that $p$ reaches after trace $t$. 
From the previous section,
we know that $t$ satisfies the invariant of Definition~\ref{def:invariant}: in
particular, we can use the states $s_{p,0}, s_{p,1}, \dots, s_{p, \ell_p}$ of
$\A$ that are given by this invariant, so that $s_{p,0}$ is the initial state
of $\A$, and $q_{p,i} \xra{\varphi_{p,i}} s_{p,i}$ for all~$i$.

To compute the tree of separation, we consider an equivalence relation $\equiv_j$ on processes defined as: $p \equiv_j p'$ if $p \equiv_{s_{r,j}} p'$ for all processes $r$. Intuitively, $p$ is made equivalent to~$p'$ if they are both potentially communicating at the end of Phase $j$, according to all the processes. If a process $r$ has fewer phases, we first normalise its local state by padding empty phases $(0, s_{r, \ell_r}, \varepsilon)$ so that in the end all processes have the same number $m$ of phases. We then proceed to the bottom-up computation, starting from the leaves. For each leaf $Y$, we synchronise the states of processes in $Y$ by applying $\synchronise\big((\sigma_{p'})_{p'\in Y}\big)$ and thus can compute an updated local state $\sigma_Y = (c_{Y,i},
q_{Y,i}, \varphi_{Y,i})_{1 \le i \le m}$, as well as a state $s_Y$ reached
in $\A$ after reading $\view_Y(t)$. For the bottom-up computation, pick a node labelled with processes~$X$ and phase $j$, with children over processes $Y_1, \dots, Y_n$ at phase $j+1$ such that $X = Y_1 \cup \cdots \cup Y_n$. Assume we have computed the joint states for each $Y_1, \dots, Y_n$, and in particular the states $s_{Y_1}, \dots, s_{Y_n}$. Consider Invariant (E) of Definition~\ref{def:invariant}. For each pair of processes $p, p' \in X$, we have $\ell_{p, p'} = j$. Therefore, for all processes in $X$, all the phases up to $j-1$ will be identical. Phase $j$ may however differ which we will reconcile as follows. 

Let $(c_i, q_i, \varphi_i)_{1\leq j\leq n}$ be the $j$-th phase of processes in  $Y_1, \ldots, Y_n$ respectively (each process of $Y_i$ share the same $i$-th phase by the synchronisation that we have first performed). Let $c = \max_{1\leq i\leq n} c_{i}$ be the maximal counter value. For all $i$, we cut $c-c_{i}$ letters to make a global synchronisation of the phase prefix (in particular the state part): thus, each local state now looks like $(c, q, \varphi'_{i})_{1 \le i \le n}$. We then reconcile $Y_1$ and $Y_2$ (continuing two by two, to reconcile all classes afterwards) as follows. First, we consider the pointwise intersection of Foata steps $\varphi'_{Y_1}$ and $\varphi'_{Y_2}$: what we have removed is necessarily independent of the other processes. Let $\varphi''$ be the result, and $s''=\delta(q, \varphi'')$. By Lemma~\ref{lem:diam}, the state reached on $\view_{Y_1 \cup Y_2}(t)$ is then $s_{Y_1\cup Y_2} = \diam(s'', s_{Y_1}, s_{Y_2}, Y_1)$.

\begin{figure}[tbp]
\centering
\scalebox{0.9}{
\begin{minipage}[c]{.45\textwidth}
\begin{tikzpicture}[greenstate/.style={rectangle, rounded corners, fill=green!30, inner sep=4pt, draw=gray, thick}, redstate/.style={rectangle, rounded corners, fill=red!30, inner sep=4pt, draw=gray, thick}, graystate/.style={rectangle, rounded corners, fill=gray!30, inner sep=4pt, draw=gray, thick}, baseline=(current bounding box.center)]
    \begin{scope}[every node/.style={greenstate}]
        \node (0) at (0,0) {\scriptsize $q_0$};
        \node (1) at (1.5, 0) {\scriptsize $q_1$};
        \node (2) at (0, -1.5) {\scriptsize $q_2$};
        \node (3) at (1.5, -1.5) {\scriptsize $q_3$};
        \node (4) at (3, -1.5) {\scriptsize $q_4$};
    \end{scope}
    \begin{scope}[every node/.style={redstate}]
        \node (5) at (4.5, -1.5) {\scriptsize $q_5$};
        \node (6) at (1.5, -3) {\scriptsize $q_6$};
        \node (7) at (3,-3) {\scriptsize $q_7$};
        \node[double] (8) at (4.5, -3) {\scriptsize $q_8$};
    \end{scope}
    \begin{scope}[->, >=stealth]
        \draw (-0.7, 0) to (0);
        \draw (0) to node [above] {\scriptsize $12$} (1);
        \draw (0) to node [left] {\scriptsize $34$} (2);
        \draw (1) to node [right] {\scriptsize $34$} (3);
        \draw (2) to node [below] {\scriptsize $12$} (3);
        \draw (3) to node [above right,xshift=-1mm,yshift=-1mm] {\scriptsize $134$} (0);
        \draw (3) to [bend right=20] node [below] {\scriptsize $13$} (4);
        \draw (4) to [bend right=20] node [above] {\scriptsize $124$} (3);
        \draw (3) to node [left] {\scriptsize $24$} (6);
        \draw (4) to node [right] {\scriptsize $24$} (7);
        \draw (5) to node [right] {\scriptsize $24$} (8);
        \draw (4) to node [above] {\scriptsize $13$} (5);
        \draw (6) to node [above] {\scriptsize $13$} (7);
        \draw (7) to node [above] {\scriptsize $13$} (8);
        \draw (5) to [loop above] node {\scriptsize $13$} (5);
        \draw (6) to [loop below] node {\scriptsize $24$} (6);
        \draw (7) to [loop below] node {\scriptsize $24$} (7);
        \draw (8) to [loop below] node {\scriptsize $13,24$} (8);
    \end{scope}
    
\draw[rounded corners,dotted] (-0.5,0.5) rectangle (3.5,-2.05);
\node() at (3.8,0) {$\Gamma_1$};
\draw[rounded corners,dotted] (4,-.5) rectangle (5,-2.05);
\node() at (5.25,-.95) {$\Gamma_2$};
\draw[rounded corners,dotted] (1,-2.45) rectangle (2,-4);
\node() at (.7,-3.4) {$\Gamma_3$};
\draw[rounded corners,dotted] (2.5,-2.45) rectangle (3.4,-4);
\node() at (3.61,-3.6) {$\Gamma_4$};
\draw[rounded corners,dotted] (4,-2.45) rectangle (5,-4);
\node() at (5.3,-3.4) {$\Gamma_5$};
\end{tikzpicture}
\end{minipage}
\hspace{1.3cm} 
\begin{minipage}[c]{.4\textwidth}
    \begin{align*}
       &p_1 \equiv_{\Gamma_1} p_2 \equiv_{\Gamma_1} p_3 \equiv_{\Gamma_1} p_4\\
       &p_1 \equiv_{\Gamma_2} p_3 \qquad p_2 \equiv_{\Gamma_2} p_4 \\ 
    &p_1 \equiv_{\Gamma_3} p_3 \qquad p_2 \equiv_{\Gamma_3} p_4 \\
           &p_1 \equiv_{\Gamma_4} p_3 \qquad p_2 \equiv_{\Gamma_4} p_4 \\ 
       &p_1 \equiv_{\Gamma_5} p_3 \qquad p_2 \equiv_{\Gamma_5} p_4
    \end{align*}
\end{minipage}
}
\caption{A connectedly communicating DFA to illustrate the need to perform intersections of views in the acceptance condition of the built AA}\label{fig:CC-DFA-intersection}
\end{figure}
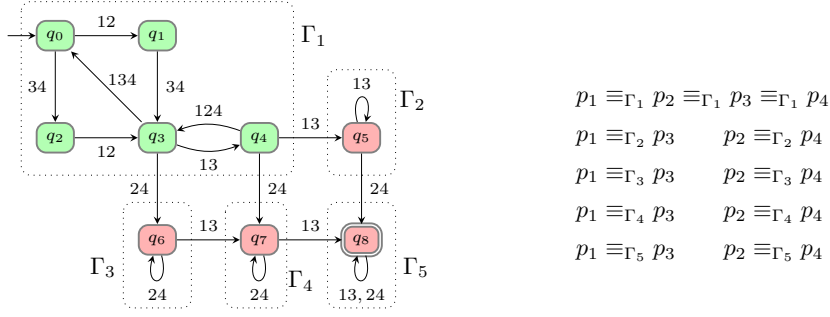

\subparagraph*{Modulo counters}

The unbounded counters make the local state space infinite. However, we can use
the crucial ingredient of Lemma~\ref{lem:pairOfProcsAtmostK-1Apart}, already
used in the construction proposed in~\cite{BerMon26} for fair DFAs, where the
views of processes cannot be more that $d$ letters apart, if $d$ is the
fairness parameter. Thus, for a set of processes $X$ communicating in an SCC of
$\A$ (that needs to be fair for $X$, by Lemma~\ref{lem:CC-fair}), the number of
letters cut by each process in $X$ lies in a contiguous range of $d$ values
where $d$ is the fairness parameter of this SCC. Thus, modulo~$2d$, these values
lie in a cyclic interval of length $d$ and when synchronising on an action, the
processes can determine the process which has cut the largest number of letters
(even if the value is lower than that of other processes) and the number of
letters to be cut to match that process. We apply this modulo counting on an
example in Appendix~\ref{app:modulo-counters-example}.

\section{Conclusion and perspectives}

In this article, we have proposed the construction of an
AA given a regular trace language that is
connectedly communicating: it is inspired by the one on fair
languages \cite{BerMon26} where only a suffix of the view is maintained by each process, but
this approach requires splitting the local states into phases to be able to cut letters from the
views of each process even in the middle. 
As perspectives, we plan on implementing the algorithm similarly as the tool
FAAST~\cite{10.5281/zenodo.18165965} proposed for the fair case. In
this tool, the construction of the AA is on-the-fly (the states and transitions are only computed when necessary at the execution) and can even be applied on non-fair languages, then providing a fair under-approximation: an
AA is built, whose language is the set of $d$-fair traces of the
specification language. We would like to study a similar strategy for 
connectedly communicating languages. 
This enables a synthesis procedure even for specifications that do not
fulfil the connected communication property, 
\todo{Ben: I have commented out the example of server-clients here, but we can reincorporate it if you feel it is important to put it in the conclusion. I have also changed the rest of the sentence to make it more precise}
thus allowing for synthesising resilient 
distributed algorithms, even without fairness properties handwritten in the formal specification.

\newpage

\appendix

\section{Additional examples}
\label{app:additional-examples}

\begin{example} \label{eg:aa-zielonka}
    Let $\Sigma = \{a, b, c\}$ be an alphabet distributed to two processes $P_1$ and $P_2$ as follows: $\Sigma_{P_1} = \{a, c\}$ and $\Sigma_{P_2} = \{b, c\}$. So the only synchronizing action is $c$. Consider the language $L = ((a + b)(a + b))^*c$. A DFA recognizing $L$ is shown in Figure~\ref{fig:DFA'}. An AA recognizing $L$ is shown in Figure~\ref{fig:DFA'-asynchronous}. Two blue $c$ or two red $c$ transitions can synchronise together. A blue $c$ and a red $c$ cannot be taken together. The transition syntax of AA allows for such coupled transitions: formally, $\delta_c((0,0)) = (2,2)$ and $\delta_c((1,1)) = (2,2)$. 
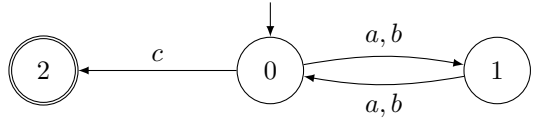
\begin{figure}
    \centering
    \begin{tikzpicture}[node distance=3cm, initial text=,>=latex]
      \node[state, initial, initial where=above, node distance=1.5cm] (0) {$0$};
      \node[state, right of=0] (1) {$1$};
      \node[state, accepting, left of=0] (2) {$2$};

      \draw[->] (0) edge node[above] {$c$} (2)
      (0) edge[bend left=10] node[above] {$a,b$} (1)
      (1) edge[bend left=10] node[below] {$a,b$} (0);
    \end{tikzpicture}
    \caption{A DFA recognising $((a+b)(a+b))^*c$}
    \label{fig:DFA'}
  \end{figure}

  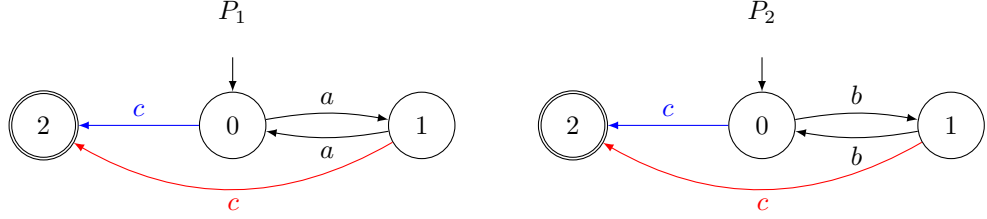
\begin{figure}
    \centering
    \begin{tikzpicture}[node distance=2.5cm, initial text=,>=latex]
      \node(P1) {$P_1$};
      \node[state, below of=P1, initial, initial where=above, node distance=1.5cm] (0) {$0$};
      \node[state, right of=0] (1) {$1$};
      \node[state, accepting, left of=0] (2) {$2$};

      \draw[->] (0) edge [blue] node[above] {$c$} (2)
      (0) edge[bend left=10] node[above] {$a$} (1)
      (1) edge[bend left=10] node[below] {$a$} (0)
      (1) edge[bend left=30, red] node[below] {$c$} (2);
      
      \node[right of=P1, node distance=7cm](P2) {$P_2$};
      \node[state, below of=P2, initial, initial where=above, node distance=1.5cm] (0') {$0$};
      \node[state, right of=0'] (1') {$1$};
      \node[state, accepting, left of=0'] (2') {$2$};

      \draw[->] (0') edge [blue] node[above] {$c$} (2')
      (0') edge[bend left=10] node[above] {$b$} (1')
      (1') edge[bend left=10] node[below] {$b$} (0')
      (1') edge[bend left=30, red] node[below] {$c$} (2');
    \end{tikzpicture}
    \caption{An AA recognizing $((a+b) (a +b ))^*c$.}    
    \label{fig:DFA'-asynchronous}
  \end{figure}
\end{example}

\begin{example}\label{ex:CC-construction}
    We now exemplify the new construction on the DFA $\A'$ of Figure~\ref{fig:ccp-automaton}. The construction is similar as in Example~\ref{ex:fair-construction} when we read letters $12$, $34$, $235$, $34$, $12$, $235$ that stay in the SCC $\Gamma_1$ where all processes potentially communicate. We continue reading the next letters in the trace of Example~\ref{ex:trace}. We start by reading $12$, $34$ (still in the SCC $\Gamma_1$), and then $13$ where we jump to the SCC $\Gamma_2$:
    \begin{center}\renewcommand{\arraystretch}{1.3}
    \arrayrulecolor{gray}
    \scriptsize
    \begin{tabular}{c  | l | l | l }
    &  $12$ & $34$ & $13$ \\
    \hline 
    $p_1$ &  $(\blue{3}, \blue{q_0}, \{12, \blue{34}\} \{\blue{235}\} \{\blue{12}\})$ & $(3, q_0, \{12, 34\} \{235\} \{12\})$ &  
    $(\blue{6}, \blue{q_0},
    \{12, \blue{34}\} \{ \blue{13} \})$ \\
    \hline
    $p_2$  & $(3, q_0, \{12, 34\} \{235\} \{\blue{12}\})$ &  $(3, q_0, \{12, 34\} \{235\} \{12\})$ & $(3, q_0, \{12, 34\} \{235\} \{12\})$ \\
    \hline
    $p_3$ &  $(3, q_0, \{12, 34\} \{235\})$ & $(3, q_0, \{12, 34\} \{235\} \{ \blue{34} \})$ &  
    $(\blue{6}, \blue{q_0},
    \{\blue{12}, 34\} \{ \blue{13} \})$ \\
    \hline
    $p_4$  & $\{12, 34\} \{235\} \{34\}$ &    $(\blue{3}, \blue{q_0}, \{\blue{12}, 34\} \{\blue{235}\} \{ \blue{34} \})$ &   $(3, q_0, \{12, 34\} \{235\} \{34\})$\\
    \hline
    $p_5$ & $(3, q_0, \{12, 34\} \{235\})$ & $(3, q_0, \{12, 34\} \{235\})$ & $(3, q_0, \{12, 34\} \{235\})$
    \end{tabular}
    \end{center}
    Processes $p_1$ and $p_3$ synchronise and expand as before, moreover cutting the two first steps, since they know every process knows them. 
    We next read letter $35$:
    \begin{center}
        \renewcommand{\arraystretch}{1.3}
        \arrayrulecolor{gray}
        \scriptsize
            \begin{tabular}{c | l }
                & $35$ \\
                \hline
                $p_1$   &  $(6, q_0, \{12, 34\}\{ 13 \})$  \\
                \hline 
                $p_2$  & $(3, q_0, \{12, 34\} \{235\} \{12\})$   \\
                \hline 
                $p_3$  & $(6, q_0, \{12, 34\} \{ 13 \} \{\blue{35}\})~\to~ (6, q_0, \{12, 34\})
                (\blue{1, q_4,} \{35\})$ \\
                \hline 
                $p_4$ & $(3, q_0, \{12, 34\} \{235\} \{34\})$   \\
                \hline 
                $p_5$  & $(\blue{6}, \blue{q_0}, \{\blue{12}, \blue{34}\} \{ \blue{13} \} \{ \blue{35} \})~\to~ (6, q_0, \{12, 34\})
                (\blue{1}, \blue{q_4}, \{35\})$
            \end{tabular}
    \end{center}
    Processes $p_3$ and $p_5$ synchronise, cutting what can be cut, and then expanding by adding $35$ in a new step. They know that before the letter $13$, they were in $\Gamma_1$, while they are in $\Gamma_2$ afterwards, with a refined potential communication relation (they know they will no longer communicate with $p_2$ and $p_4$): moreover, in $\Gamma_2$, a cut can be performed since all steps till the event labelled $13$ are $\sock[\{p_1, p_3, p_5\}]$ for them. 
    \todo{arnab: updated to mention all previous steps are $\sock$ as well.}
    We thus split the state into two phases, a first one for what is in $\Gamma_1$, and a second for what is in $\Gamma_2$: the counter of the new phase is initialised to $1$ (since we have cut one letter), and the state $q_4$ is obtained by reading the cut letter $13$ from $q_3$ (the last state reached in $\Gamma_1$, that can be computed from the state $q_0$ while reading the letters of the first phase). 
    We then read letter $24$: 
    \begin{center}
        \renewcommand{\arraystretch}{1.3}
        \arrayrulecolor{gray}
        \scriptsize
        \begin{tabular}{c |  l }
            &  $24$ \\
            \hline
            $p_1$ & $(6, q_0, \{12, 34\} \{ 13 \})$  \\
            \hline 
            $p_2$ &   $(3, q_0, \{12, 34\} \{235\} \{12, \blue{34}\} \{\blue{24}\})~\to~ (\blue{6}, \blue{q_0}, \{12,34\})
            (\blue{1, q_6, \varepsilon})$  \\
            \hline 
            $p_3$  & $(6, q_0, \{12, 34\}) 
            (1, q_4, \{35\})$\\
            \hline 
            $p_4$  & $(3, q_0, \{12, 34\} \{235\} \{12, \blue{34}\} \{\blue{24}\})~\to~ (\blue{6}, \blue{q_0}, \{12,34\})
            (\blue{1, q_6, \varepsilon})$  \\
            \hline 
            $p_5$ & $(6, q_0, \{12, 34\}) 
            (1, q_4, \{35\})$
        \end{tabular}
    \end{center}
    Processes $p_2$ and $p_4$ learn about the new SCC $\Gamma_3$, where they are now separated from every other process: they can cut the letter $24$ known by everyone they are still communicating with, and are thus also splitting into two phases. When later actions $2$ and $4$ are performed, the processes $p_2$ and $p_4$ can independently cut all of them since they are now separated from everyone. At the end of reading the whole trace of Example~\ref{ex:trace}, we reach the local states:
    \begin{center}
        \renewcommand{\arraystretch}{1.3}
        \arrayrulecolor{gray}
        \scriptsize
        \begin{tabular}{c | l}
            $p_1$ &  $(6, q_0, \{12, 34\})
            (2, q_5, \{13\})$ \\ \hline 
            $p_2$ &  $(6, q_0, \{12, 34\})
            (1, q_6, \varepsilon)$ \\ 
            \hline 
            $p_3$ & $(6, q_0, \{12, 34\})
            (2, q_5, \{13\})$ \\
            \hline 
            $p_4$ & $(6, q_0, \{12, 34\})
            (1, q_6, \varepsilon)$ \\
            \hline 
            $p_5$ & $(6, q_0, \{12, 34\})
            (1, q_4, \{35\})$
        \end{tabular}
    \end{center}
\end{example}

\begin{example}\label{ex:final-state-computation}
    In our previous example, finding if we must accept or not is non-trivial, since in the actual trace read in the DFA, we reach the accepting state $q_7$, even if no process is alone able to detect that in the AA. In this example, we first synchronise a last time processes $p_1$, $p_3$, and $p_5$ (since $p_5$ had a less recent view of the trace in the last phase). We get as new local states:  
    \begin{center}
        \renewcommand{\arraystretch}{1.3}
        \arrayrulecolor{gray}
        \scriptsize
        \begin{tabular}{c | l}
            $p_1, p_3, p_5$ &  $(6, q_0, \{12, 34\})
             (2, q_5, \{13\})$ \\ \hline 
            $p_2$ &  $(6, q_0, \{12, 34\})
            (1, q_6, \varepsilon)$\\\hline
            $p_4$ &  $(6, q_0, \{12, 34\})
            (1, q_6, \varepsilon)$
        \end{tabular}
    \end{center}
    Finally, in reverse topological order, we compute the state reached in $\A$ after having read $\view_{X}(t)$ for some subsets of processes $X$. In the example, we know that at the end of $\view_{\{p_1, p_3, p_5\}}(t)$, state $q_4$ is reached, while at the end of $\view_{\{p_2\}}(t)$ and $\view_{\{p_4\}}(t)$, state $q_6$ is reached. Moreover, at the end of the first phase, every process is aware\footnote{In the complete procedure described afterwards, separated processes may disagree on this state, which requires considering pointwise intersection of Foata steps, in order to only keep the common information, see Example~\ref{ex:acc-intersection}.} of the actual state in $\A$, that is $q_3$. Thanks to Lemma~\ref{lem:diam}, we can compute $\diam(q_3, q_4, q_6, \{p_1, p_3, p_5\})= q_7$ and declare that $q_7$ is the state of $\A$ reached after $t$. Reconciling $p_4$ does not change the state $q_7$, and we thus declare the global state of the AA accepting.
\end{example}

\section{\texorpdfstring{Computing the $\diam$ function}{Computing the diam function}}
\label{app:diam}
Assuming the existence and uniqueness of the $\diam$ function, one can compute it for each set of process $X \subseteq P$ and triple of states $(q_1, q_2, q_3)$ by finding words $u, v$ such that $\loc(u) \subseteq X$, $\loc(v) \cap X = \emptyset$, $\delta(q_1, u) = q_2$, and $\delta(q_1, v) = q_3$. If these words $u, v$ exist, then these can be found in time linear in the number of states by a simple graph search using transitions labelled by only $\Sigma_X$ letters for $u$ and only $\Sigma \setminus \Sigma_X$ letters for $v$. By the uniqueness of the $\diam$ function, we set $\diam(q_1, q_2, q_3, X) := \delta(q_1, uv)$.

\section{\texorpdfstring{Proofs of \Cref{sec:ccp-basic}}{Proofs}}
\label{app:proof-of-ccp-basic}
\PCSameInSCC*
\begin{proof}
    We show that the graph of the relation $\equiv_s$ is included in the one of $\equiv_{s'}$, and conclude by symmetry.
    Recall from \Cref{def:potential-communication-relation} that the relation $\equiv_s$ is the reflexive transitive closure of the relation $\equiv_s'$. It thus suffices to show that for all $p,p'$, if $p \equiv_s' p'$, then $p \equiv_{s'}' p'$. If $p \equiv_s' p'$, there is a word $w$ that can be read from the state $s$ where $p, p'$ are not separated. Moreover, as the state $s'$ is in the same SCC as $s$, there is a word $u$ such that $\delta(s', u) = s$. Since in the word $uw$ read from $s'$, $p$ and $p'$ are not separated, we have $p \equiv_{s'}' p'$.
\end{proof}

\todo{Changed the proof since \Cref{def:potential-communication-relation} has changed}
\CCfair*
\begin{proof}
    If $\A$ is connectedly communicating, then consider any SCC $\Gamma_i$. Let $X$ be an equivalence class of $\equiv_{\Gamma_i}$. Suppose that $\Gamma_i$ is not fair with respect to $X$, i.e.~that $\Gamma_i$ contains a cycle around a state $s\in \Gamma_i$, reading a word $v$, where some process $p\in X$ participates, but not another process $p'\in X$. In particular, $p\equiv_{\Gamma_i} p'$.

    We assume there is a path starting from $s$ labelled $w$ such that $p$ and $p'$ are not separated in $w$, i.e., $p \equiv_{\Gamma_i}' p'$ as in \Cref{def:potential-communication-relation}. Let $s'$ be the final state of this path.

    Since $\A$ is trim, there exists a word $u$ that reaches $s$ from the initial state, and a word $w'$ that reaches an accepting state from $s'$. Moreover, there exists a word $u'$ that reaches $s'$ from $s$. In particular, for every $k \geq 0$, $uu' v^k ww'$ is a witness that $\Lt(\A)$ is not $k$-connectedly communicating since $|v^k|_p \geq k$, $|v^k|_{p'} = 0$, and $p$ and $p'$ are not separated in $ww'$.

    In particular, $p'$ participates in $v$, and $p$ and $p'$ are not separated in $v$. Since the relation $\equiv_{\Gamma_i}$ is the reflexive transitive closure of $\equiv_{\Gamma_i}'$, it follows that for all $p, p' \in X$, $p \equiv_{\Gamma_i} p'$.
    
    \medskip 
    Reciprocally, suppose that for every SCC $\Gamma_i$, and every equivalence class $X\subseteq P$ of $\equiv_{\Gamma_i}$, the SCC $\Gamma_i$ is fair with respect to $X$. Consider a trace $t\in \Lt(\A)$, two processes $p, p'$ and a linearisation $uvw$ with $|v|_p \geq n$, and $|v|_{p'} =0$, with $n$ being the number of states of $\A$.  We need to show that $p$ and $p'$ are separated in $w$. 

    Since $t$ is accepted, there are states $q, q', q_f$ such that $q_0 \xrightarrow{u} q \xrightarrow{v} q' \xrightarrow{w} q_f$ with $q_0$
    initial and $q_f$ accepting.
    Since $|v|_p\geq n$, there exists a state $r$ such that the subrun $q \xrightarrow{v} q'$ visits $r$ at least twice, with $p$ participating at least once in-between. Let $v_1 v_2 v_3$ be the corresponding factorisation, such that $r\xrightarrow{v_2}r$ and $|v_2|_p > 0$.
    
    By contradiction, if $p$ and $p'$ are not separated in $w$, then by \Cref{def:potential-communication-relation}, this is a witness that $p\equiv_{\Gamma_i} p'$ for $\Gamma_i$, the SCC of state $q'$. Since there is a path from $r$ to $q'$, this implies that $p\equiv_{\Gamma_j} p'$ with $\Gamma_j$ the SCC of state $r$. But we have $r\xrightarrow{v_2}r$ with $|{v_2}|_p \geq k$, $|{v_2}|_{p'} = 0$. This shows that the SCC $\Gamma_j$ is not fair with respect to $[p]_{\Gamma_j} = [p']_{\Gamma_j}$.
\end{proof}

\section{\texorpdfstring{Proofs of \Cref{sec:ccp-construction}}{Proofs}}
\label{app:proofs-ccp-construction}


\begin{restatable}{lemma}{CutSplitPreservesInv}
    \label{lem:cut-split-preserves-invariant}
    For any trace $t$ the tuple $(\dvdcut_{p, q_0}(\view_p(t)))_{p \in P}$ satisfies the conditions of \Cref{def:invariant}.
\end{restatable}

\begin{proof}
    It is clear from the definition of $\dvdcut$ operation that conditions $\mathrm{(A)}, \mathrm{(B)}$ and $\mathrm{(C)}$ are satisfied. We need only show that conditions $\mathrm{(D)}$ and $\mathrm{(E)}$ are satisfied.

    Consider a pair of processes $p$ and $p'$ and the pair of tuples $(c_{p,i}, q_{p,i}, \varphi_{p,i})_{1\leq i\leq \ell_p}$ and $(c_{p',i}, q_{p',i}, \varphi_{p',i})_{1\leq i\leq \ell_{p'}}$ where $p'\in [p]_{s_{p,\ell_p}}$, $p \in [p']_{s_{p',\ell_{p'}}}$, and for all~$i$, $q_{p,i} \xrightarrow{\varphi_{p,i}} s_{p,i}$ and $q_{p',i} \xrightarrow{\varphi_{p',i}} s_{p',i}$. 
    
    We assume for a contradiction that $\mathrm{(D)}$ is violated for the pair of processes $p$ and $p'$ that are still communicating, i.e., $\ell_p = \ell_{p'} + 2$ (w.l.o.g.) process $p$ has cut and split two more phases than $p'$. From condition $\mathrm{(C)}$ we have $\Foata(\view_p(t)) = \psi_{p,1} \varphi_{p,1} \psi_{p,2} \varphi_{p,2} \cdots \psi_{p,\ell_p} \varphi_{p,\ell_p}$ and $\Foata(\view_{p'}(t)) = \psi_{p',1} \varphi_{p',1} \psi_{p',2} \varphi_{p',2} \cdots \psi_{p',\ell_{p'}} \varphi_{p',\ell_{p'}}$ for some $\psi_{p,i}$, $\psi_{p',i}$ with $|\psi_{p,i}| = c_{p, i}$ and $|\psi_{p', i}| = c_{p', i}$. We have two cases:
    \begin{itemize}
        \item Process $p'$ has cut fewer letters in the $\ell_{p'}^{th}$ phase than $p$, i.e., $c_{p', \ell_{p'}} < c_{p, \ell_{p'}}$. Note that the Foata steps $\psi_{p, \ell_p - 1}$ were cut for process $p$ since every step till then are $\sock[[p]_{q_{\ell_p-1}}]$ for $p$. In particular all steps of $\Foata(\view_p(t))$ till the steps $\psi_{p, \ell_p - 1}$ are a prefix of $\Foata(\view_{p'}(t))$. Consider some event $e$ in the steps preceding $\psi_{p, \ell_p - 1}$ that is not known by some process $r$ in $\view_{p'}(t)$ and hence $\psi_{p, \ell_p - 1}$ was not cut in $\Foata(\view_{p'}(t))$. This event was however in $r$'s view of $\psi_{p,1} \varphi_{p,1} \psi_{p,2} \varphi_{p,2} \cdots \psi_{p,\ell_p-1} \varphi_{p,\ell_p-1}$, i.e., there is an $r$ event $e_r$ and a $p$ event $e_p$ in $\varphi_{p,\ell_p-1}$ such that $e \le e_r \le e_p$. In particular, $e_r$ does not occur in $\varphi_{p', \ell_{p'}}$.
        
        Since events in $\psi_{p, \ell_p}$ were also cut for process $p$, by definition of $\dvdcut$ every event of $\varphi_{p, \ell_p-1}$ occurs in the view of every process in $[p]_{s_{p, \ell_p}}$, and thus in $\varphi_{p', \ell_{p'}}$. This means $e_r$ also occurs in $\varphi_{p', \ell_{p'}}$, a contradiction.
        
        \item Process $p$ has cut fewer letters than process $p'$ in the $\ell_{p'}^{th}$ phase. Then, there is an event $e$ in $\psi_{p', \ell_{p'}}$ that is not $\sock[[p]_{s_{p, \ell_{p'}}}]$ for $p$ due to some process $r$ but appears in $\varphi_{p, \ell_{p'}}$. However, there is an $r$-event $e_r$ and a $p'$-event $e_{p'}$ in $\psi_{p', \ell_{p'}}\varphi_{p', \ell_{p'}}$ such that $e \le e_r \le e_{p'}$. In the joint view of $p$ and $p'$, $\Foata(\view_{p,p'}(t))$, the event $e_{p'}$ must appear in a Foata step corresponding to a step in $\varphi_{p, \ell_p}$, for all $p'$ events preceding it are known to $p$ by definition of $\dvdcut$. However, this implies that in $p$'s last phase, $r$ and $p'$ are still communicating with each other, and therefore $p$ and $r$ are still communicating by transitivity. However then the phase $(c_{p, \ell_p-1}, q_{p, \ell_p-1}, \varphi_{p, \ell_p-1})$ shouldn't have been initiated since $r$ communicates with $p$ in phase $\ell_p-1$ but does not know about the event~$e$.
    \end{itemize}


    As a direct consequence, for every pair of processes $p$ and $p'$ that are still communicating such that $\ell_p = \ell_{p'} + 1$, the first $\ell_{p'}-1$ phases of $p$ and $p'$ are the same since the views of $p$ and $p'$ are the same till the steps $\psi_{\ell_{p'}}$. Therefore, condition $\mathrm{(E)}$ holds.
\end{proof}

\begin{restatable}{lemma}{CutSplitState}
    \label{lem:states-cut-split-view}
    The local state of each process $p$ of $\B$ after reading a trace $t$ is $\dvdcut(\view_p(t))$.
\end{restatable}

\begin{proof}
    We prove by induction on $t$. The base case $t = \epsilon$ is trivial. We assume the lemma holds for a trace $t$ and prove that it holds for a trace $ta$ for some letter $a$, i.e., we assume the global state of $\B$ after reading $t$ is $(\dvdcut(\view_p(t))_{p \in P}) = \big((c_{p,i}, q_{p,i}, \varphi_{p,i})_{1\leq i\leq \ell_p}\big)_{p \in P}$. For $p \notin \loc(a)$, we have $\dvdcut(\view_p(ta)) = \dvdcut(\view_p(t))$ since $\view_p(ta) = \view_p(t)$. We consider the local states of processes in $\loc(a)$.
    
    
    Let $\ell = \max_{p \in \loc(a)} \ell_p$, and $p_{max}$ be a process such that $\ell_{p_{max}} = \ell$. Also, let $c = \max_{p \in \loc(a)}(c_{\ell_p})$. Note by invariant $\mathrm{E}$, the first $\ell_p - 2$ phases are the same for all $p \in \loc(a)$. By invariant $\mathrm{C}$, we let $\Foata(\view_p(t)) = \psi_{p,1} \varphi_{p,1} \psi_{p,2} \varphi_{p,2} \cdots \psi_{p,\ell_p} \varphi_{p,\ell_p}$ for some $\psi_{p,i}$ with $|\psi_{p,i}| = c_{p, i}$. Thus, for all $p, p' \in \loc(a)$, for all $1 \le i < \ell$, we have $\psi_{p, i} = \psi_{p', i}$ and $\varphi_{p, i} = \varphi_{p', i}$. Thus, after splitting the last phase of processes $p$ with $\ell_p = \ell - 1$ to match the penultimate phase of $p_{max}$, processes in $\loc(a)$ have the first $\ell-1$ phases of $\dvdcut(\view_{\loc(a)}(ta))$. Therefore, synchronising the last phase as in the fair construction, expanding by $a$, and performing the final $\dvdcut$ operation yields the final phases of $\dvdcut(\view_{\loc(a)}(ta))$.

\end{proof}

\paragraph*{Formalizing the accepting state computation}
Assume we are given a global state $\sigma = (\sigma_p)_{p \in \proc}$. Each $\sigma_p$ is of this form:
\begin{align*}
\sigma_p = (c_{p,i}, q_{p,i}, \varphi_{p,i})_{1 \le i \le \ell_p}
\end{align*}
Process $p$ has $\ell_p$ phases. As done before, let $s_{p, i}$ be the state reached by $\A$ on reading $\varphi_{p,i}$ starting from $q_{p,i}$. 

First we do a normalization so that all the processes have the same number of phases. Let $m_{\sigma} = \max_{p \in \proc} \ell_p$. If for some process $p'$, we have $\ell_{p'} < m_\sigma$, then we add $m_\sigma - \ell_{p'}$ extra phases $(0, s_{p', \ell_{p'}}, \varepsilon)$. For these extra phases, that is, for $\ell_{p'} + 1 \le j \le m_\sigma$, we also have $s_{p', j} = s_{p', \ell_{p'}}$. Therefore, we now have a global state of the form:
\begin{align*}
 \sigma_p = (c_{p,i}, q_{p,i}, \varphi_{p,i})_{1 \le i \le m_{\sigma}}
\end{align*}

Define a bunch of $m_{\sigma}$ equivalence relations on the set of processes as follows: for $1 \le j \le m_{\sigma}$, we say $p \equiv^j p'$ if for all processes $r \in \proc$,  we have $p \equiv_{s_{r, j}} p'$.

The following lemma is a consequence of the fact that for every process $r$, the equivalence relation $\equiv_{s_{r, j+1}}$ is a refinement of $\equiv_{s_{r, j}}$.

\begin{lemma}
    \label{lem:enhanced-state-equiv-refinement}
    The relation $\equiv^{j+1}$ is a refinement of $\equiv^j$ for all $1 \le j \le m_{\sigma} - 1$. 
\end{lemma}



This gives the partitions of processes communicating with each other in each ``global'' phase. The division of the processes looks like this:

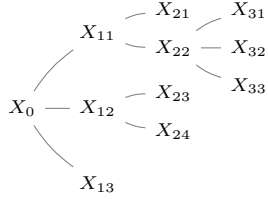
\begin{tikzpicture}
\node (0) at (0,0) {\scriptsize $X_0$};
\node (11) at (1, 1) {\scriptsize $X_{11}$};
\node (12) at (1, 0) {\scriptsize $X_{12}$};
\node (13) at (1, -1) {\scriptsize $X_{13}$};
\node (21) at (2, 1.3) {\scriptsize $X_{21}$};
\node (22) at (2, 0.8) {\scriptsize $X_{22}$};
\node (23) at (2, 0.2) {\scriptsize $X_{23}$};
\node (24) at (2, -0.3) {\scriptsize $X_{24}$};
\node (31) at (3, 1.3) {\scriptsize $X_{31}$};
\node (32) at (3, 0.8) {\scriptsize $X_{32}$};
\node (33) at (3, 0.3) {\scriptsize $X_{33}$};

\begin{scope}[gray]
\draw (0) to [bend left=10] (11);
\draw (0) to (12);
\draw (0) to [bend right=10] (13);
\draw (11) to [bend left=10] (21);
\draw (11) to [bend right=10] (22);
\draw (12) to [bend left=10] (23);
\draw (12) to [bend right=10] (24);
\draw (22) to [bend left=10] (31);
\draw (22) to (32);
\draw (22) to [bend right=10] (33);
\end{scope}
\end{tikzpicture}

We can now define a sequence of \emph{enhanced states} $\eta_{m_\sigma + 1}, \eta_{m_\sigma}$, $\eta_{m_\sigma - 1}, \dots, \eta_1$ which we will define inductively.

We start with $\eta_{m_\sigma + 1}$. Assume $\equiv^{m_\sigma + 1}$ is the trivial relation where each process is the only one its equivalence class. Define $\eta_{m_\sigma + 1}$ to be $(\sigma_p ~\mid~ s_{p, m_\sigma})_{p \in P}$.

Suppose we have $\eta_{j+1}$: to each equivalence class $Y$ of $\equiv^{j+1}$, we have associated:
\begin{align*}
 (c_{Y,1}, q_{Y,1}, \varphi_{Y,1}) \cdots (c_{Y, j} q_{Y, j}, \varphi_{Y, j}) ~\mid~ s_Y
\end{align*}
where $(c_{Y,1}, q_{Y,1}, \varphi_{Y,1}) \cdots (c_{Y, j}, q_{Y, j}, \varphi_{Y,j})$ equals $(c_{p,1}, q_{p,1}, \varphi_{p,1}) \cdots (c_{p,j}, q_{p,j}, \varphi_{p,j})$ for some $p \in Y$. No matter which $p\in Y$ we pick, this part will be the same. 

Consider an equivalence class $X$ of $\equiv^j$. By \Cref{lem:enhanced-state-equiv-refinement}, there is a collection of equivalence classes $\mathcal{Y} = \{Y_i\}_{1 \le i \le n}$ of $\equiv^{j+1}$ such that $X = \bigcup \mathcal{Y}$.

By invariant $\mathrm{(E)}$ of \Cref{def:invariant}, we have that for all $Y, Y' \in \mathcal{Y}$ and $0 \le \ell \le j - 1$, $(c_{Y, \ell}, q_{Y, \ell}, \varphi_{Y, \ell}) = (c_{Y', \ell}, q_{Y', \ell}, \varphi_{Y', \ell})$ since processes in $X$ communicate with each other in phase $j$. We synchronise the enhanced states of $\{Y_i\}_{1 \le i \le n}$ as follows.
\begin{enumerate}
    \item Let $c_j := \max_{y \in \mathcal{Y}} c_{Y, j}$. For each $1 \le i \le n$ cut $c_{Y, j} - c_{Y_i, j}$ letters of $\varphi_{Y_i, j-1}$ to get the phase $(c_{Y, j}. q_{Y, j}, \varphi'_{Y_i, j})$.
    \item \label{enum:final-state-for-loop-init} Let $\sigma_Z = (c_{X, \ell}, q_{X, \ell}, \varphi_{X, \ell})_{1 \le \ell \le j} := \sigma_{Y_1}$ and $s_Z := s_{Y_1}$.
    \item \label{enum:final-state-for-loop} For $i := 2$ to $n$,
    \begin{enumerate}
        \item \label{enum:final-state-for-loop-i} Let $\varphi_{ZY_i}$ be the Foata step-wise intersection of $\varphi_{Z, j}$ and $\varphi'_{Y_i, j}$, and $q_{ZY_i} = \delta(q_{Y,j}, \varphi_{XY_i})$ be the state reached on reading the intersection. The remainder of the views of $Z$ and $Y_i$ starting from $q_{ZY_i}$ are independent. 
        \item \label{enum:final-state-for-loop-ii} Let $s_Z := \diam(q_{ZY_i}, s_Z, s_{Y_i}, \cup_{1 \le \ell < i} Y_\ell)$, and $\varphi_{Z, j}$ be the Foata step-wise union of $\varphi_{Z, j}$ and $\varphi'_{Y_i, j}$.
    \end{enumerate}
    \item Finally define $\eta_{j, X} := (c_{Z, \ell}, q_{Z, \ell}, \varphi_{Z, \ell})_{1 \le \ell < j} ~\mid~ s_Z$
\end{enumerate}

\begin{lemma}
    \label{lem:final-state-computation}
    If $t$ is the trace read by $\B$, then for each $1 \le j \le m_{\sigma}+1$ and equivalence class $X$ of $\equiv^j$, $s_X$ is the state reached by $\A$ on reading $\view_X(t)$ where $(c_{X, i}, q_{X, i}, \varphi_{X, i})_{1 \le i < j} \mid s_X$ is the enhanced state of $X$ in $\eta_j$.
\end{lemma}

\begin{proof}
    We prove by a bottom-up induction on the enhanced phases. The base case of the $m_{\sigma}+1$ phase follows from invariant $\mathrm{(C)}$ of \Cref{def:invariant}. We assume the invariant holds for each equivalence class $Y$ of $\equiv^{j+1}$. Consider some equivalence class $X$ of $\equiv^j$ and set of equivalence classes $\mathcal{Y} = \{Y_i\}_{1 \le y \le n}$ of $\equiv^{j+1}$ such that $X = \bigcup \mathcal{Y}$. Further let the enhanced state of each $Y_i \in \mathcal{Y}$ be $(c_h, q_h, \varphi_h)_{1 \le h \le j} \mid s_i$.
    
    We prove by induction on the for loop of line \ref{enum:final-state-for-loop} that at the end of iteration $i$, $s_Z$ is the state reached in $\A$ on reading $\view_{Y_1, \dots Y_i}(t)$. The base case when $Z = \{Y_1\}$ (line \ref{enum:final-state-for-loop-init}) follows by the earlier inductive hypothesis. We assume the hypothesis holds at the end of iteration $i$. Since the views of processes $Z$ and $Y_{i+1}$ are the same up to the Foata steps preceding $\varphi_{Z, j}$ and $\varphi_{Y_{i+1}, j}$ (by invariants $\mathrm{(C)}$ and $\mathrm{(E)}$ of \Cref{def:invariant}), $q_{ZY_{i+1}}$ is the state of $\A$ reached on reading the intersection of $\view_{Z}(t)$ and $\view_{Y_{i+1}}(t)$. Since, the remainder of the views of $Z$ and $Y_{i+1}$ starting from $q_{ZY_i}$ are independent of each other, the $\diam$ function of line \ref{enum:final-state-for-loop-ii} is well-defined, and gives the state reached on reading $\view_{Z \cup \{Y_{i+1}\}}(t)$.
\end{proof}

\section{The example of the main body with modulo counters}
\label{app:modulo-counters-example}
We reconsider the DFA of Figure~\ref{fig:ccp-automaton}, but reading a longer trace than in Example~\ref{ex:CC-construction} to demonstrate what happens when we use modulo counters, in the finite AA that we are constructing. We first read the five same first letters:

{
\renewcommand{\arraystretch}{1.3}
\arrayrulecolor{gray}
\scriptsize 
\[\begin{array}{c|l|l|l|l|l} 
    &  12 & 34 & 235 & 34 & 12 \\
    \hline
    p_1 & (0, q_0, \{\textcolor{blue}{12}\}) & (0, q_0, \{12\}) & (0, q_0, \{12\}) & (0, q_0, \{12\}) &  (0, q_0, \{12, \blue{34}\} \{\blue{235}\} \{\blue{12}\})\\
    \hline 
    p_2 & (0, q_0, \{\textcolor{blue}{12}\}) & (0, q_0, \{12\}) & (0, q_0, \{ 12, \textcolor{blue}{34}\}\{\textcolor{blue}{235}\}) & (0, q_0, \{12, 34\} \{235\}) & (0, q_0, \{12, 34\} \{235\} \{\blue{12}\}) \\
    \hline 
    p_3 & (0, q_0, \varepsilon) & (0, q_0, \{\textcolor{blue}{34}\}) & (0, q_0, \{\textcolor{blue}{12}, 34\}\{\textcolor{blue}{235}\}) & (0, q_0, \{12, 34\} \{235\} \{\blue{34}\}) & (0, q_0, \{12, 34\} \{235\} \{34\})\\
    \hline
    p_4 & (0, q_0, \varepsilon) & (0, q_0, \{\textcolor{blue}{34}\}) & (0, q_0, \{34\}) & (0, q_0, \{\blue{12}, 34\} \{\blue{235}\} \{\blue{34}\}) & (0, q_0, \{12, 34\} \{235\} \{34\}) \\
    \hline 
    p_5 & (0, q_0, \varepsilon) & (0, q_0, \varepsilon) & (0, q_0, \{\textcolor{blue}{12,34}\} \{\textcolor{blue}{235}\}) & (0, q_0, \{12,34\} \{235\}) & (0, q_0, \{12, 34\} \{235\})
\end{array}\]
}

We continue as before by reading $235\; 12\; 34$ where we start cutting steps, and thus incrementing the counters: 

{
\renewcommand{\arraystretch}{1.3}
\arrayrulecolor{gray}
\scriptsize 
\[\begin{array}{c|l|l|l} 
    & 235 & 12 & 34 \\\hline
    p_1 & (0, q_0, \{12, 34\} \{235\} \{12\}) &
(\blue{3}, \blue{q_0},  \{12, \blue{34}\} \{\blue{235}\} \{\blue{12}\}) & (3, q_0, \{12, 34\} \{235\} \{12\}) \\
    \hline
    p_2 & (\blue{3}, \blue{q_0}, \{12, 34\} \{235\}) & 
 (3, q_0, \{12, 34\} \{235\} \{\blue{12}\}) &  (3, q_0, \{12, 34\} \{235\} \{12\})  \\
    \hline
    p_3 & (\blue{3}, \blue{q_0}, \{12, 34\} \{235\}) & 
  (3, q_0, \{12, 34\} \{235\}) & (3, q_0, \{12, 34\} \{235\} \{ \blue{34} \}) 
  \\
    \hline
    p_4 & (0, q_0, \{12, 34\} \{235\} \{34\}) & 
    (0, q_0, \{12, 34\} \{235\} \{34\}) & (\blue{3}, \blue{q_0}, \{\blue{12}, 34\} \{\blue{235}\} \{ \blue{34} \}) \\
    \hline
    p_5 & (\blue{3}, \blue{q_0}, \{12, 34\} \{235\}) & 
  (3, q_0, \{12, 34\} \{235\}) & (3, q_0, \{12, 34\} \{235\})
\end{array}\]
}

We continue by reading new letters: $235\; 34\; 12\; 235$. On the last one, the counter should be increased to $9$, but since we count modulo $2d=8$ (since the SCC $\Gamma_1$ is fair with a fairness parameter $d=4$), the counter becomes $1$:

{
\renewcommand{\arraystretch}{1.3}
\arrayrulecolor{gray}
\scriptsize 
\[\begin{array}{c|l|l|l|l} 
    & 235 & 34 & 12 & 235\\\hline
    p_1 & (3, q_0, \{12, 34\} \{235\} \{12\}) & (3, q_0, \{12, 34\} \{235\} \{12\}) & (\blue{6}, \blue{q_0}, \{12, 34\} \{235\}\{\blue{12}\}) & (6, q_0, \{12, 34\} \{235\}\{12\}) \\\hline
    p_2 & (\blue{6}, \blue{q_0}, \{12, \blue{34}\} \{\blue{235}\}) & (6, q_0, \{12, 34\} \{235\}) & (6, q_0, \{12, 34\} \{235\}\{\blue{12}\})& (\blue{1}, \blue{q_0}, \{12, \blue{34}\}\{\blue{235}\})\\\hline
    p_3 & (\blue{6}, \blue{q_0}, \{\blue{12}, 34\} \{\blue{235}\})& (6, q_0, \{12, 34\} \{235\} \{\blue{34}\}) & (6, q_0, \{12, 34\} \{235\} \{34\})& (\blue{1}, \blue{q_0}, \{\blue{12}, 34\}\{\blue{235}\})\\\hline
    p_4 & (3, q_0, \{12, 34\} \{235\} \{34\}) & (\blue{6}, \blue{q_0}, \{12, 34\} \{235\} \{\blue{34}\}) & (6, q_0, \{12, 34\} \{235\} \{34\}) & (6, q_0, \{12, 34\} \{235\} \{34\})\\\hline
    p_5 & (\blue{6},\blue{q_0}, \{\blue{12}, \blue{34}\} \{\blue{235}\}) & (6, q_0, \{12, 34\} \{235\}) & (6, q_0, \{12, 34\} \{235\}) & (\blue{1}, \blue{q_0}, \{\blue{12}, \blue{34}\}\{\blue{235}\})
\end{array}\]
}

We now read letters $12\; 34$; The synchronisation between $p_1$ and $p_2$ (and similarly for $p_3$ and $p_4$) is between counter values $6$ and $1$: since counter values of potentially communicating processes cannot be more than $d=4$ apart (by Lemma~\ref{lem:pairOfProcsAtmostK-1Apart}), the counter $1$ is necessary ahead of $6$ even if $1<6$. Processes $p_2$ and $p_4$ thus know that they are late and can cut the corresponding number of letters to catch up. 

{
\renewcommand{\arraystretch}{1.3}
\arrayrulecolor{gray}
\scriptsize 
\[\begin{array}{c|l|l} 
        & 12 & 34 \\\hline
    p_1 & (\blue{1}, \blue{q_0}, \{12, \blue{34}\}\{\blue{235}\}\{\blue{12}\}) 
    & (1, q_0, \{12, 34\}\{235\}\{12\})\\\hline
    p_2 & (1, q_0, \{12, 34\}\{235\}\{\blue{12}\}) & (1, q_0, \{12, 34\}\{235\}\{12\}) \\\hline
    p_3 & (1, q_0, \{12, 34\}\{235\}) & (1, q_0, \{12, 34\}\{235\}\{\blue{34}\})\\\hline
    p_4 & (6, q_0, \{12, 34\} \{235\} \{34\}) & (\blue{1}, \blue{q_0}, \{\blue{12}, 34\}\{\blue{235}\}\{\blue{34}\})\\\hline
    p_5 & (1, q_0, \{12, 34\}\{235\}) & (1, q_0, \{12, 34\}\{235\})
\end{array}\]
}

We continue by reading letters that exit the SCC $\Gamma_1$, that are $24\; 13\; 35$. The SCC $\Gamma_2$ is fair with a delay constant $2$, and we thus count modulo $2\times 2-1 = 3$, while the SCC $\Gamma_3$ is fair with a delay constant $1$ for each of the participating processes, and we thus count modulo $2\times 1-1=1$ which means that the counter value will always stay at $0$ (as expected since participating processes $p_2$ and $p_4$ are now separated from all other processes). When reading $24$, we already cut in $\Gamma_3$, and thus create a new phase right away. For the SCC $\Gamma_2$, we need to wait for the letter $35$ to be able to cut the letter $13$ and thus create a new phase.

{
\renewcommand{\arraystretch}{1.3}
\arrayrulecolor{gray}
\scriptsize 
\[\begin{array}{c|l|l|l} 
        & 24 & 13 & 35 \\\hline
    p_1 & (1, q_0, \{12, 34\}\{235\}\{12\}) & (1, q_0, \{12, 34\}\{235\}\{12,\blue{34}\}\{\blue{13}\}) & (1, q_0, \{12, 34\}\{235\}\{12,34\}\{13\})\\\hline
    p_2 & (1, q_0, \{12, 34\}\{235\}\{12,\blue{34}\})(\blue{0},\blue{q_6},\blue{\varepsilon}) & (1, q_0, \{12, 34\}\{235\}\{12,34\})(0,q_6,\varepsilon)  & (1, q_0, \{12, 34\}\{235\}\{12,34\})(0,q_6,\varepsilon) 
    \\\hline
    p_3 & (1, q_0, \{12, 34\}\{235\}\{34\})  &  (1, q_0, \{12, 34\}\{235\}\{\blue{12},34\}\{\blue{13}\})& (1, q_0, \{12, 34\}\{235\}\{12,34\})(\blue{1},\blue{q_4},\{\blue{35}\}) \\\hline
    p_4 & (1, q_0, \{12, 34\}\{235\}\{\blue{12},34\})(\blue{0},\blue{q_6},\blue{\varepsilon})& (1, q_0, \{12, 34\}\{235\}\{12,34\})(0,q_6,\varepsilon) & (1, q_0, \{12, 34\}\{235\}\{12,34\})(0,q_6,\varepsilon) \\\hline
    p_5 & (1, q_0, \{12, 34\}\{235\})& (1, q_0, \{12, 34\}\{235\})& (1, q_0, \{12, 34\}\{235\}\{\blue{12},\blue{34}\})(\blue{1},\blue{q_4},\{\blue{35}\})
\end{array}\]
}

On the next letter $13$ that we read, process $p_1$ learns about the phase cutting and that he needs to cut step $\{13\}$: after expansion, the step $\{35\}$ is known by the three processes $p_1, p_3, p_5$ and is thus cut. When reading the last letter $13$, process $p_1$ has counter value $2$ and process $p_3$ has counter value $0$: since processes are at most $2-1=1$ counter value apart, they know that $p_3$ is ahead and process $p_1$ can thus catch up.

{
\renewcommand{\arraystretch}{1.3}
\arrayrulecolor{gray}
\scriptsize 
\[\begin{array}{c|l|l|l} 
        & 13 & 35 & 13\\\hline
    p_1 &  (1, q_0, \{12, 34\}\{235\}\{12,34\})(\blue{2},\blue{q_5},\{\blue{13}\}) &  (1, q_0, \{12, 34\}\{235\}\{12,34\})(2,q_5,\{13\}) & (1, q_0, \{12, 34\}\{235\}\{12,34\})(\blue{1},\blue{q_5},\{\blue{13}\})\\\hline
    p_2 & (1, q_0, \{12, 34\}\{235\}\{12,34\})(0,q_6,\varepsilon) & (1, q_0, \{12, 34\}\{235\}\{12,34\})(0,q_6,\varepsilon) & (1, q_0, \{12, 34\}\{235\}\{12,34\})(0,q_6,\varepsilon) 
    \\\hline
    p_3 & (1, q_0, \{12, 34\}\{235\}\{12,34\})(\blue{2},\blue{q_5},\{\blue{13}\}) & (1, q_0, \{12, 34\}\{235\}\{12,34\})(\blue{0},\blue{q_4},\{\blue{35}\}) & (1, q_0, \{12, 34\}\{235\}\{12,34\})(\blue{1},\blue{q_5},\{\blue{13}\}) \\\hline
    p_4 & (1, q_0, \{12, 34\}\{235\}\{12,34\})(0,q_6,\varepsilon) & (1, q_0, \{12, 34\}\{235\}\{12,34\})(0,q_6,\varepsilon) & (1, q_0, \{12, 34\}\{235\}\{12,34\})(0,q_6,\varepsilon) 
    \\\hline
    p_5 & (1, q_0, \{12, 34\}\{235\}\{12,34\})(1,q_4,\{35\}) & (1, q_0, \{12, 34\}\{235\}\{12,34\})(\blue{0},\blue{q_4},\{\blue{35}\}) & (1, q_0, \{12, 34\}\{235\}\{12,34\})(0,q_4,\{35\})
\end{array}\]
}

When letter $2$ and $4$ are read, nothing is changing in the local states of processes $p_2$ and $p_4$ respectively, since every letter can now be cut with a counter value staying $0$.

In the global state that we reach, processes $p_1, p_3, p_5$ that are still potentially communicating, synchronise for a last time and get the global state:
\begin{center}\renewcommand{\arraystretch}{1.3}
        \arrayrulecolor{gray}
        \scriptsize
        \begin{tabular}{c|l}
    $p_1, p_3, p_5$ & $(1, q_0, \{12, 34\}\{235\}\{12,34\})(1,q_5,\{13\}) $\\\hline
    $p_2$ & $(1, q_0, \{12, 34\}\{235\}\{12,34\})(0,q_6,\varepsilon)$ \\\hline
    $p_4$ & $(1, q_0, \{12, 34\}\{235\}\{12,34\})(0,q_6,\varepsilon)$
\end{tabular}
\end{center}
To reconcile $p_1$ (and $p_3, p_5$) with $p_2$, we use the common knowledge of the state reached at the end of phase $1$, that is $q_3$, and the knowledge that $p_1$ ended in $q_4$ after the last phase, while $p_2$ ended in $q_6$: since $\diam(q_3, q_4, q_6, \{p_1, p_2\}) = q_7$, and the fact that reconciling with $p_4$ does not change anything, we declare the global state accepting.


\end{document}